\documentclass{article}
\pdfoutput=1

\usepackage{arxiv}
\usepackage[utf8]{inputenc} 
\usepackage[T1]{fontenc}    
\usepackage[hidelinks]{hyperref}       
\usepackage{url}            
\usepackage{booktabs}       
\usepackage{nicefrac,mathtools}       
\usepackage{microtype}      
\usepackage{xcolor,comment}         
\usepackage{fullpage}
\usepackage{amsfonts,amsmath,amssymb,amsthm}
\usepackage{graphicx}
\usepackage{subcaption}
\usepackage{enumitem}
\usepackage[numbers]{natbib}
\usepackage{mathtools}
\usepackage{fancyhdr}
\usepackage{bm}
\usepackage{tikz}
\usetikzlibrary{
    arrows,
    calc,
    positioning,
    intersections,
}

\newtheorem{claim}{Claim}

\newtheorem{proposition}{Proposition}

\newtheorem{definition}{Definition}
\newtheorem{lemma}{Lemma}
\newtheorem{theorem}{Theorem}
\newtheorem*{theorem*}{Theorem}
\newtheorem*{question*}{Question}
\newtheorem*{informal_theorem}{Informal Theorem}
\newcommand{\email}[1]{\href{mailto:#1}{\texttt{#1}}}

\DeclareMathOperator*{\argmax}{arg\!\,max}
\DeclareMathOperator*{\argmin}{arg\!\,min}
\DeclarePairedDelimiter{\norm}{\lVert}{\rVert} 
\newtheorem{corollary}{Corollary}

\title{Fast Convergence of Optimistic Gradient Ascent in Network Zero-Sum Extensive Form Games}

\author{%
  Georgios Piliouras\\
  Singapore Univ. of Technology and Design\\
  Singapore\\
  \email{georgios@sutd.edu.sg} \\
  \And
  Lillian Ratliff\\
  University of Washington\\
  Seattle, Washington\\
  \email{ratliffl@uw.edu}\\
  \And
  Ryann Sim\\
  Singapore Univ. of Technology and Design\\
  Singapore\\
  \email{ryann\_sim@mymail.sutd.edu.sg} \\
  \And
  Stratis Skoulakis\\
  EPFL\\
  Laussane, Switzerland\\
  \email{efstratios.skoulakis@epfl.ch}
}
\date{}
\renewcommand{\headeright}{}
\renewcommand{\undertitle}{}
\renewcommand{\shorttitle}{Fast Convergence of OGA in Network Zero-Sum EFGs}

\newcommand{\la}{\langle}
\newcommand{\ra}{\rangle}
\newcommand{\X}{\mathcal{X}}
\newcommand{\N}{\mathcal{N}}
\newcommand{\A}{\mathcal{A}}
\newcommand{\V}{\mathcal{V}}
\newcommand{\I}{\mathcal{I}}
\newcommand{\Y}{\mathcal{Y}}
\newcommand{\G}{\mathcal{G}}
\newcommand{\Z}{\mathcal{Z}}
\renewcommand{\P}{\mathcal{P}}
\renewcommand{\H}{\mathcal{H}}
\newcommand{\dist}{\mathrm{dist}}
\newcommand{\boxnew}[1]{\ \boxed{#1}\ }

\newcommand\blfootnote[1]{%
  \begingroup
  \renewcommand\thefootnote{}\footnote{#1}%
  \addtocounter{footnote}{-1}%
  \endgroup
}

\begin{document}

\maketitle
\blfootnote{This is the full version of a paper to appear in the 15th Symposium on Algorithmic Game Theory (SAGT 2022)}
\begin{abstract}
The study of learning in games has thus far focused primarily on normal form games.  
In contrast, our understanding
 of learning in \emph{extensive form games} (EFGs) and particularly in EFGs with many agents lags far behind, despite them being closer in nature to many real world applications. We consider the natural class of \textit{Network Zero-Sum Extensive Form Games}, which combines the global zero-sum property of agent payoffs, the efficient representation of graphical games as well the expressive power of EFGs. 
%
We examine the convergence properties of \textit{Optimistic Gradient Ascent} (OGA) in these games. We prove that the time-average behavior of such online learning dynamics exhibits $O(1/T)$ rate convergence to the set of Nash Equilibria. Moreover, we show that the day-to-day behavior
also converges to Nash with rate $O(c^{-t})$ for some game-dependent constant $c>0$.
\end{abstract}

\section{Introduction}\label{sec:intro}

\textit{Extensive Form Games} (EFGs) are an important class of games which have been  studied for more than $50$
years \cite{K50}. 
EFGs capture various settings where several selfish agents sequentially perform actions which change the \textit{state of nature}, with the action-sequence finally leading to a \textit{terminal state}, at which each agent receives a payoff. The most ubiquitous examples of EFGs are real-life games such as Chess, Poker, Go etc. Recently the application of the \textit{online learning framework} has proven to be very successful in the design of modern AI which can beat even the best human players in real-life games \cite{tammelin2015solving,brown2018superhuman}. 
At the same time, online learning in EFGs has many interesting applications in economics, AI, machine learning and sequential decision making that extend far beyond the design of game-solvers \cite{AB16,perolat2021poincare}.  

Despite its numerous applications, online learning in EFGs is far from well understood. From a practical point of view, testing and experimenting with various online learning algorithms in EFGs requires a huge amount of computational resources due to the large number of states in EFGs of interest \cite{zinkevich2007regret,rowland2019multiagent}. 
From a theoretical perspective, it is known that online learning dynamics may oscillate, cycle or even admit chaotic behavior even in very simple settings \cite{palaiopanos2017multiplicative,mertikopoulos2018cycles,leonardos2022exploration}. 
On the positive side, there exists a recent line of research into the special but fairly interesting class of \textit{two-player zero-sum EFGs}, which provides the following solid claim: \textit{In two-player zero-sum EFGs, the time-average strategy vector produced by online learning dynamics converges to the Nash Equilibrium (NE), while there exist online learning dynamics which exhibit day-to-day convergence}
\cite{farina2019optimistic,wei2020linear,wei2021last}. 
Since in most settings of interest there are typically multiple interacting agents, the above results motivate the following question:
\begin{question*}
Are there natural and important classes of multi-agent extensive form games for which online learning dynamics converge to a Nash Equilibrium? Furthermore, what type of convergence is possible? Can we only guarantee time-average convergence, or can we also prove day-to-day convergence (also known as last-iterate convergence) of the dynamics?
\end{question*}

In this paper we answer the above questions in the positive for an interesting class of multi-agent EFGs called \textit{Network Zero-Sum Extensive Form Games}. A Network EFG consists of a graph $\G=(V,E)$ where each vertex $u \in V$ represents a selfish agent and each edge $(u,v) \in E$ corresponds to an extensive form game $\Gamma^{uv}$ played between the agents $u,v \in V$. Each agent $u \in V$ selects her strategy so as to maximize the overall payoff from the games corresponding to her incident edges. The game is additionally called \textit{zero-sum} if the sum of the agents' payoffs is equal to zero no matter the selected strategies.  

We analyze the convergence properties of the online learning dynamics produced when all agents of a Network Zero-Sum EFG update their strategies according to \textit{Optimistic Gradient Ascent}, and show the following result:
\begin{informal_theorem}
When the agents of a network zero-sum extensive form game update their strategies using Optimistic Gradient Ascent, their time-average strategies converge with rate $O(1/T)$ to a Nash Equilibrium, while the last-iterate mixed strategies converge to a Nash Equilibrium with rate $O(c^{-t})$ for some game-dependent constant $c>0$.
\end{informal_theorem}

Network Zero-Sum EFGs are an interesting class of multi-agent EFGs for much the same reasons that network zero-sum normal form games are interesting, with several additional challenges. Indeed, due to the prevalence of networks in computing systems, there has been increased interest in network formulations of normal form games \cite{kearns2013graphical}, which have been applied to multi-agent reinforcement learning \cite{yang2020overview} 
and social networks \cite{kempe2003maximizing}.

Network Zero-Sum EFGs can be seen as a natural model of closed systems in which selfish agents compete over a fixed set of resources \cite{daskalakis2009network,cai2011minmax}, thanks to their global constant-sum property\footnote{equivalent to the global zero-sum property.} (the edge-games are not necessarily zero-sum).
For example, consider the users of an online poker platform playing \textit{Heads-up Poker}, a two-player extensive form game. Each user can be thought of as a node in a graph and two users are connected by an edge (corresponding to a poker game) if they play against each other. 
Note that here, each edge/game differs from another due to the differences in the dollar/blind equivalence. Each user $u$ selects a poker-strategy to utilize against the other players, with the goal of maximizing her overall payoff. This is an indicative example which can clearly be modeled as a Network Zero-Sum EFG.

In addition, Network Zero-Sum EFGs are also attractive to study due to the fact that their descriptive complexity scales \emph{polynomially} with the number of agents. Multi-agent EFGs that cannot be decomposed into pairwise interactions (i.e., do not have a network structure) admit an exponentially large description with respect to the number of the agents \cite{kearns2013graphical}. 
Hence, by considering this class of games, we are able to exploit the decomposition to extend results that are known for network normal form games to the extensive form setting.

\textbf{Our Contributions.}
To the best of our knowledge, this is the first work establishing convergence to Nash Equilibria of online learning dynamics in network extensive form games with more than two agents. As already mentioned, there has been a stream of recent works establishing the convergence to Nash Equilibria of online learning dynamics in two-player zero-sum EFGs. However, there are several key differences between the two-player and the network cases. All the previous works concerning the two-player case follow a \textit{bilinear saddle point approach}. Specifically, due to the fact that in the two-agent case any Nash Equilibrium coincides with a min-max equilibrium, the set of Nash Equilibria can be expressed as the solution to the following bilinear saddle-point problem:
\[\mathrm{min}_{x \in \mathcal{X}}
\mathrm{max}_{y \in \mathcal{Y}}~  x^\top \cdot A\cdot y  =  \mathrm{max}_{y \in \mathcal{Y}} \mathrm{min}_{x \in \mathcal{X}}~x^\top \cdot A\cdot y\]
Any online learning dynamic or algorithm that converges to the solution of the above saddle-point problem thus also converges to the Nash equilibrium in the two-player case.

However, in the network setting, there is no min-max equilibrium and hence no such connection between the Nash Equilibrium and saddle-point optimization. To overcome this difficulty, we establish that Optimistic Gradient Ascent in a class of EFGs known as \emph{consistent} Network Zero-Sum EFGs (see Section \ref{sec:setting}) can be equivalently described as optimistic gradient descent in a \textit{two-player symmetric game $(R,R)$} over a \textit{treeplex polytope $\mathcal{X}$}. We remark that both the matrix $R$ and the treeplex polytope $\mathcal{X}$ are constructed from the Network Zero-Sum EFG. Using the zero-sum property of Network EFGs, we show that the constructed matrix $R$ satisfies the following `restricted' zero-sum property:
\begin{equation}\label{eq:zero-sum}
x^\top  \cdot R \cdot y + y^\top  \cdot R \cdot x = 0\text{   for all }x,y \in \mathcal{X}    
\end{equation}
Indeed, Property~(\ref{eq:zero-sum}) is a generalization of the \textit{classical zero-sum property} $A=-A^\top$. In general, the constructed matrix $R$ does not satisfy $R = -R^\top$ and Property~(\ref{eq:zero-sum}) simply ensures that the sum of payoffs equal to zero only when $x,y \in \mathcal{X}$. Our technical contribution consists of generalizing the analysis of \cite{wei2020linear} (which holds for classical two-player zero-sum games) to symmetric games satisfying Property~(\ref{eq:zero-sum}).


\textbf{Related Work.}
\textit{Network Zero-Sum Normal Form Games} \cite{daskalakis2009network,cai2011minmax,CCDP16} are a special case of our setting, where each edge/game is a normal form game. Network zero-sum normal form games present major complications compared to their two-player counterparts. The most important of these complications is that in the network case, there is no min-max equilibrium. In fact, different Nash Equilibria can assign different values to the agents. All the above works
study linear programs for computing Nash Equilibria in network zero-sum normal form games. \cite{cai2011minmax} introduce the idea of connecting a network zero-sum normal form game with an equivalent symmetric game $(R,R)$ which satisfies Property~(\ref{eq:zero-sum}). This generalizes the linear programming approach of two-player zero-sum normal form games to the network case. They also show that in network normal form zero-sum games, the time-average behavior of online learning dynamics converge with rate $\Theta(1/\sqrt{T})$ to the Nash Equilibrium.

The properties of \textit{online learning in two-player zero-sum EFGs} have been studied extensively in literature.
\cite{zinkevich2007regret} and \cite{lanctot2009monte} propose no-regret algorithms for extensive form games with $O(1/\sqrt{T})$ average regret and polynomial running time in the size of the game. More recently, regret-based algorithms achieve $O(1/T)$ time-average convergence to the min-max equilibrium \cite{HGPS10,kroer2018solving,farina2019optimistic} for \textit{two-player zero-sum EFGs}. Finally, \cite{lee2021last} and \cite{wei2020linear} establish that Online Mirror Descent achieves $O(c^{-t})$ last-iterate convergence (for some game-dependent constant $c \in (0,1)$) in \textit{two-player zero-sum EFGs}.

\section{Preliminaries}\label{sec:preliminaries}


\subsection{Two-Player Extensive Form Games}
\begin{definition}\label{d:EFG}
A two-player extensive form game $\Gamma$ is a tuple $\Gamma := \left <\H,\A,\Z, p , \I\right>$ where 
\begin{itemize}
    \item $\H$ denotes the states of the game that are decision points for the agents. The states $h \in \H$ form a tree rooted at an initial state $r \in \H$.

    \item Each state $h \in \H$ is associated with a set of \textit{available actions} $\A(h)$.
    
    \item Each state $h \in \mathcal{H}$ admits a label $\mathrm{Label}(h) \in \{1,2,c\}$ denoting the \textit{acting player} at state $h$. The letter $c$ denotes a special agent called a \textit{chance agent}. Each state $h \in \H$ with $\mathrm{Label}(h)= c$ is additionally associated with a function $\sigma_h: \A(h)\mapsto [0,1]$ where $\sigma_h(\alpha)$ denotes the probability that the chance player selects action $\alpha \in \A(h)$ at state $h$, $\sum_{\alpha\in \A(h)}\sigma_h(\alpha) =1$.

    \item $\mathrm{Next(\alpha,h)}$ denotes the state $h' := \mathrm{Next(\alpha,h)}$ which is reached when agent $i:= \mathrm{Label}(h)$ takes action $\alpha \in \A(h)$ at state $h$. $\H_i \subseteq \H$ denotes the states $h \in \H$ with $\mathrm{Label}(h) = i$.  
    
    \item $\Z$ denotes the terminal states of the game corresponding to the leaves of the tree. At each $z \in \Z$ no further action can be chosen, so $\A(z) = \varnothing$ for all $z \in \Z$. Each terminal state $z\in \Z$ is associated with values $(u_1(z),u_2(z))$ where $p_i(z)$ denotes the payoff of agent $i$ at terminal state $z$.
    
    \item Each set of states $\H_i$ is further partitioned into \textit{information sets} $(\I_1,\ldots,\I_k)$ where $\I(h)$ denotes the information set of state $h \in \H_i$. In the case that $\I(h_1) = \I(h_2)$ for some $h_1,h_2 \in \H_1$, then $\A(h_1) = \A(h_2)$.
\end{itemize}
\end{definition}

\textit{Information sets} model situations where the \textit{acting agent} cannot differentiate between different states of the game due to a lack of information. Since the agent cannot differentiate between states of the same information set, the available actions at states $h_1, h_2$ in the same information set $(\I(h_1) = \I(h_2))$ must coincide, in particular $\A(h_1) = \A(h_2)$.
\begin{definition}\label{d:behavioral_plan}
A behavioral plan $\sigma_i$ for agent $i$ is a function such that for each state $h \in \H_i$, $\sigma_i (h)$ is a probability distribution over $\A(h)$ i.e. $\sigma_i(h,\alpha)$ denotes the probability that agent $i$ takes action $\alpha \in \A(h)$ at state $h \in \H_i$. Furthermore it is required that $\sigma_i(h_1) = \sigma_i(h_2)$ for each $h_1,h_2 \in \H_i$ with $\I(h_1) = \I(h_2)$.
The set of all behavioral plans for agent $i$ is denoted by $\Sigma_i$.
\end{definition}
The constraint $\sigma_i(h_1) = \sigma_i(h_2)$ for all $h_1,h_2 \in \H_i$ with $\I(h_1) = \I(h_2)$ models the fact that since agent $i$ cannot differentiate between states $h_1,h_2$, agent $i$ must act in the exact same way at states $h_1,h_2 \in \H_i$.
\begin{definition}
For a collection of behavioral plans $\sigma = (\sigma_1,\sigma_2) \in \Sigma_1 \times \Sigma_2$ the payoff of agent $i$, denoted by $U_i(\sigma)$, is defined as:
\[U_i(\sigma) := \sum_{z \in \Z} p_i(z) \cdot \underbrace{\Pi_{(h,h') \in \P(z)} ~\sigma_{\mathrm{Label(h)}}(h,\alpha_{h'})}_{\text{probability that state $z$ is reached}} \]
where $\P(z)$ denotes the path from the root state $r$ to the terminal state $z$ and $\alpha_{h'}$ denotes the action $\alpha \in \H_i$ such that $h' = \mathrm{Next}(h,\alpha)$. 
\end{definition}
\begin{definition}\label{d:nash_beh}
A collection of behavioral plans $\sigma^\ast = (\sigma_1^\ast,\sigma^\ast_2)$ is called a Nash Equilibrium if for all agents $i = \{1,2\}$,
\[U_i(\sigma^\ast_i, \sigma^\ast_{-i}) \geq U_i(\sigma_i, \sigma^\ast_{-i})~~~\text{for all }\sigma_i \in \Sigma_i\]
\end{definition}
The classical result of \cite{nash1951non} proves the existence of Nash Equilibrium in normal form games. This result also generalizes to a wide class of extensive form games which satisfy a property called \textit{perfect recall} (\cite{kuhn1953contributions,selten1965spieltheoretische}).

\begin{definition}\label{d:perfect_recall}
A two-player extensive form game $\Gamma := \left <\H,\A,\Z, p , \I\right>$ has \textbf{perfect recall} if and only if for all states $h_1,h_2 \in \H_i$ with $\I(h_1) = \I(h_2)$ the following holds: Define the sets
$\P(h_1)\cap \H_i := (p_1,\ldots,p_k,h_1)$ and $\P(h_2)\cap \H_i := (q_1,\ldots,q_m,h_2)$. Then:
\begin{enumerate}
    \item $ k = m$.
    \item $\I(p_\ell) = \I(q_\ell)~~$ for all $\ell \in \{1,\ldots,k\}$.
    \item $p_{\ell + 1} \in \mathrm{Next}(p_\ell,\alpha,i) $ and $q_{\ell + 1} \in \mathrm{Next}(q_\ell,\alpha,i) $ for some action $\alpha \in \A(p_\ell)$ (since $\A(p_\ell) = \A(q_\ell)$). 
\end{enumerate}
\end{definition}
Before proceeding, let us further explain the perfect recall property. As already mentioned, agent $i$ cannot differentiate between states $h_1,h_2 \in \H_i$ when $\I(h_1) = \I(h_2)$. In order for the state $h_1$ to be reached, agent $i$ must take some specific actions along the path $\P(h_1)\cap \H_i := (p_1,\ldots,p_k,h_1)$. The same logic holds for $\P(h_2)\cap \H_i := (p_1,\ldots,p_k,h_1)$. In case where agent $i$ could distinguish $\P(h_1)\cap \H_i$ from set $\P(h_2)\cap \H_i$, then she could distinguish state $h_1$ from $h_2$ by recalling the previous states in $\H_i$. This is the reason for the second constraint in Definition~\ref{d:perfect_recall}. Even if $\I(p_\ell) = \I(q_\ell)~$ for all $\ell \in \{1,\ldots,k\}$, agent $i$ could still distinguish $h_1$ from $h_2$ if 
$p_{\ell + 1} \in \mathrm{Next}(p_\ell,\alpha,i)$ and $q_{\ell + 1} \in \mathrm{Next}(q_\ell,\alpha',i)$. In such a case, agent $i$ can distinguish $h_1$ from $h_2$ by recalling the actions that she previously played and checking if the $\ell$-th action was $\alpha$ or $\alpha'$. This case is encompassed by the third constraint.

\subsection{Two-Player Extensive Form Games in Sequence Form}
A \textit{two-player extensive form game} $\Gamma$ can be captured by a two-player bilinear game where the action spaces of the agents are a specific kind of polytope, commonly known as a \textit{treeplex} \cite{HGPS10}. In order to formally define the notion of a treeplex, we first need to introduce some additional notation.

\begin{definition}
Given an two-player extensive form game $\Gamma$, we define the following:
\begin{itemize}
    \item $\P(h)$ denotes the path from the root state $r \in \H$ to the state $h \in \H$. 

    \item $\mathrm{Level}(h)$ denotes the distance from the root state $r\in \H$ to state $h\in \H$.
    \item $\mathrm{Prev}(h,i)$ denotes the lowest ancestor of $h$ in the set $\H_i$. In particular, $$\mathrm{Prev}(h,i) = \mathrm{argmax}_{h' \in \P(h) \cap \H_i} \mathrm{Level}(h').$$
    \item The set of states $\mathrm{Next}(h,\alpha, i) \subseteq \H$ denotes the highest descendants $h' \in \H_i$ once action $\alpha \in \A(h)$ has been taken at state $h$. More formally, $h' \in \mathrm{Next}(h,\alpha, i)$ if and only if in the path $\P(h,h') = (h,h_1,\ldots,h_k,h')$, all states $h_\ell \notin \H_i$ and $h_1 = \mathrm{Next}(h, \alpha)$.
\end{itemize}
\end{definition}

\begin{definition}\label{d:treeplex}
Given a two-player extensive form game $\Gamma$, the set $\X^\Gamma_i$ is composed by all vectors $x_i \in [0,1]^{|\H_i|+|\Z|}$ which satisfy the following constraints:
\begin{enumerate}
    
    \item $x_i(h)= 1~$ for all $h \in \H_i$ with $\mathrm{Prev}(h,i) = \varnothing$.
    
    \item $x_i(h_1) = x_i(h_2)~$ if there exists $h_1',h_2' \in \H_i$ such that $h_1 \in \mathrm{Next}(h_1',\alpha,i)$, $h_2 \in \mathrm{Next}(h_2',\alpha,i)$ and $\I(h_1') = \I(h_2')$.
    
    \item $\sum_{\alpha \in \A(h)}x_i(\mathrm{Next}(h,\alpha,i)) = x_i(h)~~$ for all $h \in \H_i$.
\end{enumerate}
\end{definition}
A vector $x_i \in \X^{\Gamma}_i$ is typically referred to as an agent $i$'s \textit{strategy in sequence form}. Strategies in sequence form come as an alternative to the behavioral plans of Definition~\ref{d:behavioral_plan}. As established in Lemma~\ref{l:equiv}, there exists an equivalence between a behavioral plan $\sigma_i \in \Sigma_i$ and a strategy in sequence form $x_i \in \X_i^{\Gamma}$ for games with perfect recall. 
\begin{lemma}\label{l:equiv}
Consider a two-player extensive form game $\Gamma$ with perfect recall
and the $(|\H_1| + |\Z|)\times (|\H_2| + |\Z|)$ dimensional matrices $A^\Gamma_1,A^\Gamma_2$ with $[A_i^\Gamma]_{zz} = p_i(z)$ for all terminal nodes $z \in \Z$ and $0$ otherwise. There exists a polynomial-time algorithm transforming any behavioral plan $\sigma_i \in \Sigma_i$ to a vector $x_{\sigma_i} \in \X_i^\Gamma$ such that 
\[ U_1(\sigma_1,\sigma_2) = x_{\sigma_1}^\top \cdot A^\Gamma_1 \cdot x_{\sigma_2}~~~~ \text{    and   }~~~~U_2(\sigma_1,\sigma_2) = x_{\sigma_2}^\top \cdot A^\Gamma_2 \cdot x_{\sigma_1}\]
Conversely, there exists a polynomial-time algorithm transforming any vector $x_i \in \X^\Gamma_i$ to a vector $\sigma_{x_i} \in \Sigma_i$ such that
\[ x_1^\top \cdot A^\Gamma_1 \cdot x_2 = U_1(\sigma_{x_1},\sigma_{x_2})~~~~ \text{    and   }~~~~ x_2^\top \cdot A^\Gamma_2 \cdot x_1 = U_2(\sigma_{x_1},\sigma_{x_2})\]
\end{lemma}
To this end, one can understand why \textit{strategies in sequence form} are of great use. Assume that agent $2$ selects a behavioral plan $\sigma_2 \in \Sigma_2$. Then, agent $1$ wants to compute a behavioral plan $\sigma_1^\ast \in \Sigma_1$ which is the \textit{best response} to $\sigma_2$, namely 
$\sigma_1^\ast := \mathrm{argmax}_{\sigma_1 \in \Sigma_1}U_1(\sigma_1,\sigma_2)$. This computation can be done in polynomial-time in the following manner: Agent~$1$ initially converts (in polynomial time) the behavioral plan $\sigma_2$ to $x_{\sigma_2} \in \X^\Gamma_2$, which is the respective strategy in sequence form. Then, she can obtain a vector $x_1^\ast = \mathrm{argmax}_{x_1 \in \X^\Gamma_1}~x_1^\top \cdot A^\Gamma_1\cdot x_2$. The latter step can be done in polynomial-time by computing the solution of an appropriate linear program. Finally, she can convert the vector $x_1^\ast$ to a behavioral plan $\sigma_{x^\ast_1} \in \Sigma_1$ in polynomial-time. Lemma~\ref{l:equiv} ensures that $\sigma_{x^\ast_1} = \mathrm{argmax}_{\sigma_1 \in \Sigma_1}U_1(\sigma_1,\sigma_2)$.

The above reasoning can be used to establish an equivalence between the Nash Equilibrium $(\sigma_1^\ast, \sigma_2^\ast)$ of an EFG $\Gamma := \left <\H,\A,\Z, p , \I\right>$ with the Nash Equilibrium in its sequence form.  
\begin{definition}
A Nash Equilibrium of a two-player EFG $\Gamma$ in sequence form is a vector $(x_1^\ast , x_2^\ast) \in \X_1^\Gamma \times \X_2^\Gamma$ such that 
\begin{itemize}
    \item $(x^\ast_1)^\top \cdot A_1^\Gamma \cdot x^\ast_2 \geq (x_1)^\top \cdot A_1^\Gamma \cdot x^\ast_2 ~~~\text{for all }x_1 \in \X_1^\Gamma$
    
    \item $(x^\ast_2)^\top \cdot A_2^\Gamma \cdot x^\ast_1 \geq (x_2)^\top \cdot A_2^\Gamma \cdot x^\ast_1 ~~~\text{for all }x_2 \in \X_2^\Gamma$
\end{itemize}
\end{definition}
Lemma~\ref{l:equiv} directly implies that any Nash Equilibrium of an EFG $(\sigma_1^\ast, \sigma_2^\ast) \in \Sigma_1\times \Sigma_2$ as per Definition~\ref{d:nash_beh} can be converted in polynomial-time to a Nash Equilibrium in the sequence form $(x_1^\ast, x_2^\ast)\in \X_1^\Gamma  \times \X_2^\Gamma $ and vice versa.

\subsection{Optimistic Mirror Descent}
In this section we introduce and provide the necessary background for \textit{Optimistic Mirror Descent} \cite{rakhlin2013online}. For a convex function $\psi: \mathbb{R}^d \mapsto \mathbb{R}$, the corresponding \textit{Bregman divergence} is defined as
\[D_\psi(x, y) := \psi(x) - \psi(y) -  \left \langle \nabla \psi(y),x-y  \right \rangle\]
If $\psi$ is $\gamma$-strongly convex, then $D_\psi(x,y) \geq \frac{\gamma}{2} \norm{x-y}$. Here and in the rest of the paper, we note that $\|\cdot\|$ is shorthand for the $L_2$-norm.

Now consider a game played by $n$ agents, where the action of each agent $i$ is a vector $x_i$ from a convex set $\X_i$. Each agent selects its action $x_i \in \X_i$ so as to minimize her individual cost (denoted by $C_i(x_i,x_{-i})$), which is continuous, differentiable and convex with respect to $x_i$. Specifically, $$C_i(\lambda\cdot x_i + (1-\lambda)\cdot x_i',x_{-i}) \leq \lambda\cdot C_i(x_i,x_{-i})
+ (1-\lambda)\cdot C_i(x'_i,x_{-i}) \text{ for all } \lambda \in [0,1]$$

Given a step size $\eta > 0$ and a convex function $\psi(\cdot)$ (called a regularizer), \textit{Optimistic Mirror Descent} (OMD) sequentially performs the following update step for $t=1,2,\ldots$:
\begin{align}
    x_i^t &= \text{argmin}_{x \in \X_i}\left\{ \eta \left\langle x, F^{t-1}_i(x)\right\rangle + D_\psi\left(x,\hat{x}_i^t\right)\right\}\label{eq:Mirror_Descent1}\\
    \hat{x}_i^{t+1} &= \text{argmin}_{x \in \X_i}\left\{ \eta \left\langle x, F^t_i(x) \right\rangle + D_\psi(x,\hat{x}_i^t)\right\}\label{eq:Mirror_Descent2}
\end{align}
where $F_i^t(x_i) = \nabla_{x_i} C_i(x_i , x_{-i}^t)$ and $D_\psi(x,y)$ is the \textit{Bregman Divergence} with respect to $\psi(\cdot)$. If the step-size $\eta$ selected is sufficiently small, then \textit{Optimistic Mirror Descent} ensures the \textit{no-regret property} \cite{rakhlin2013online}, making it a natural update algorithm for selfish agents \cite{H16}. To simplify notation we denote the projection operator of a convex set $\X^\ast$ as $\Pi_{\X^\ast}(x):= \argmax_{x^\ast \in \X^\ast}\norm{x - x^\ast}$ and the squared distance of vector $x$ from a convex set $\X^\ast$ as $\dist^2(x,\X^\ast) := \norm{x - \Pi_{\X^\ast}(x)}^2$.

\section{Our Setting}\label{sec:setting}
In this section of the paper, we introduce the concept of Network Zero-Sum Extensive Form Games, which are a network extension of the two player EFGs introduced in Section \ref{sec:preliminaries}. 
\subsection{Network Zero-Sum Extensive Form Games}
A \textit{network extensive form game} is defined with respect to an undirected graph $\G=(V,E)$ where nodes $V$ $\left(\vert V \vert =n\right)$ correspond to the set of players and each edge $(u,v) \in E$ represents a \textit{two-player extensive form game} $\Gamma^{uv}$ played between agents $u,v$. Each node/agent $u \in V$ selects a behavioral plan $\sigma_u \in \Sigma_u$ which they use to play all the \textit{two-player EFGs} on its outgoing edges.
\begin{definition}[Network Extensive Form Games]
A network extensive form game is a tuple $\Gamma:= \left<\G,\H,\A,\Z,\I\right>$ where  
\begin{itemize}
    \item $\G = (V,E)$ is an an undirected graph where the nodes $V$ represents the agents.
    
    \item Each agent $u \in V$ admits a set of states $\H_u$ at which the agent $u$ plays. Each state $h \in \H_u$ is associated with a set $\A(h)$ of possible actions that agent $u$ can take at state $h$.
    
    \item $\mathrm{\I}(h)$ denotes the information set of $h \in \H_u$. If  $~\I(h)=\I(h')$ for some $h,h' \in \H_u$ then $\A(h) = \A(h')$.
    
    \item For each edge $(u,v) \in E$, $\Gamma^{uv}$ is a two-player extensive form game with perfect recall. The states of $\Gamma^{uv}$ are denoted by $\H^{uv} \subseteq \H_u \cup \H_v$.
    
    \item For each edge $(u,v) \in E$, $\Z^{uv}$ is the set of terminal states of the two-player extensive form game $\Gamma^{uv}$ where $p^{\Gamma_{uv}}_u(z)$
    denotes the payoffs of $u,v$ at the terminal state $z \in \Z^{uv}$. The overall set of terminal states of the network extensive form game is the set $\Z := \cup_{(u,v)\in E}\Z^{uv}$.
\end{itemize}
\end{definition}

In a network extensive form game, each agent $u\in V$ selects a behavioral plan $\sigma_u \in \Sigma_u $ (see Definition~\ref{d:behavioral_plan}) that they use to play the two-player EFG's $\Gamma^{uv}$ with $(u,v)\in E$. Each agent selects her behavioral plan so as to maximize the sum of the payoffs of the two-player EFGs in her outgoing edges.
\begin{definition}
Given a collection of behavioral plans $\sigma = (\sigma_1,\ldots, \sigma_n) \in \Sigma_1 \times \ldots \times \Sigma_n$ the payoff of agent $u$, denoted by $U_u(\sigma)$, equals
\[U_u(\sigma) := \sum_{v:(u,v)\in E}p_{u}^{\Gamma_{uv}}(\sigma_u,\sigma_v) \]
Moreover a collection $\sigma^\ast = (\sigma^\ast_1,\ldots, \sigma^\ast_n) \in \Sigma_1 \times \ldots \times \Sigma_n$ is called a Nash Equilibrium if and only if
\[ U_u(\sigma^\ast_u,  \sigma^\ast_{-u}) \geq U_u(\sigma_u,  \sigma^\ast_{-u})~~~~\text{for all }\sigma_u \in \Sigma_u\]
\end{definition}
As already mentioned, each agent $u \in V$ plays all the two-player games $\Gamma^{uv}$ for $(u,v)\in E$ with the same behavioral plan $\sigma_u \in \Sigma_u$. This is due to the fact that the agent cannot distinguish between a state $h_1,h_2 \in \H_u$ with $\I(h_1) = \I(h_2)$ even if $h_1,h_2$ are states of different EFG's $\Gamma^{uv}$ and $\Gamma^{uv'}$. As in the case of \textit{perfect recall}, the latter implies that $u$ cannot differentiate states $h_1,h_2$ even when recalling the states $\H_u$ visited in the past and her past actions. In Definition~\ref{d:consistency} we introduce the notion of \textit{consistency} (this corresponds to the notion of \textit{perfect recall} for two-player extensive form games (Definition~\ref{d:perfect_recall})). From now on we assume that the network EFG is consistent without mentioning it explicitly.
\begin{definition}\label{d:consistency}
A network extensive form game $\Gamma:= \left<\G,\H,\A,\Z,\I\right>$ is called \textbf{consistent} if and only if for all players $u \in V$ and states $h_1, h_2 \in \H_u$ with $\I(h_1) = \I(h_2)$ the following holds: for any $(u,v),(u,v') \in E$ the sets
$\P^{uv}(h_1)\cap \H_u := (p_1,\ldots,p_k,h_1)$ and $\P^{uv'}(h_2)\cap \H_u := (q_1,\ldots,q_m,h_2)$ satisfy:
\begin{enumerate}
    \item $k = m$.
    \item $\I(p_\ell) = \I(q_\ell)~~$ for all $\ell \in \{1,k\}$.
    
    \item $p_{\ell + 1} \in \mathrm{Next}^{\Gamma_{uv}}(p_\ell,\alpha,u) $ and $q_{\ell + 1} \in \mathrm{Next}^{\Gamma_{uv'}}(q_\ell,\alpha,u) $ for some action $\alpha \in \A(p_\ell)$. 
\end{enumerate}
where $\P^{uv}(h)$ denotes the path from the root state to state $h$ in the two-player extensive form game $\Gamma^{uv}$.
\end{definition}

In this work we study the special class of network \emph{zero-sum} extensive form games. This class of games is a generalization of the network zero-sum normal form games studied in \cite{cai2011minmax}.

\begin{definition}\label{d:zero_sum}
A behavioral plan $\sigma_u \in \Sigma_u$ of Definition~\ref{d:behavioral_plan} is called pure if and only if $\sigma_u(h,\alpha)$ either equals $0$ or $1$ for all actions $\alpha \in \A(h)$. A network extensive form game is called \textbf{zero-sum} if and only if for any collection $\sigma:= (\sigma_1,\ldots,\sigma_n)$ of pure behavioral plans, $U_u(\sigma) = 0~~~\text{for all } u\in V$.
\end{definition}

\subsection{Network Extensive Form Games in Sequence Form}
As in the case of two-player EFGs, there exists an equivalence between behavioral plans $\sigma_u \in \Sigma_u$ and strategies in sequence form $x_u$. As we shall later see, this equivalence is of great importance since it allows for the design of natural and computationally efficient learning dynamics that converge to Nash Equilibria both in terms of behavioral plans and strategies in sequence form.
\begin{definition}\label{d:treeplex_graph}
Given a network extensive form game $\Gamma:= \left<\G,\H,\A,\Z,\I\right>$, the treeplex polytope $\X_u \subseteq [0,1]^{|\H_u| + |\Z_u|}$ is the set defined as follows: $x_u \in \X_u$ if and only if
\begin{enumerate}
    \item $x_u \in \X^{\Gamma_{uv}}_u~~$ for all $(u,v)\in E$.
    
    \item $x_u(h_1) = x_u(h_2)$ in case there exists $(u,v),(u,v')\in E$ and $h_1',h_2' \in \H_u$ with $\I(h'_1)=I(h'_2)$ such that $h_1 \in \mathrm{Next}^{\Gamma_{uv}}(h_1',\alpha,u)$, $h_2 \in \mathrm{Next}^{\Gamma_{uv'}}(h_2',\alpha,u)$ and $\I(h_1')=I(h_2')$.
\end{enumerate}
\end{definition}
The second constraint in Definition~\ref{d:treeplex_graph} is the equivalent of the second constraint in Definition~\ref{d:treeplex}. We remark that the linear equations describing the treeplex polytope $\X_u$ can be derived in polynomial-time with respect to the description of the network extensive form game. In Lemma~\ref{l:equiv_graph} we formally state and prove the equivalence between behavioral plans and strategies in sequence form. 
\begin{lemma}\label{l:equiv_graph}
    Consider the matrix $A^{uv}$ of dimensions $(|\H_u| + |\Z^u|) \times (|\H_v| + |\Z^v|)$ such that 
    \[ \left[A^{uv}\right]_{h_1h_2} = \left\{
\begin{array}{ll}
      p_u^{\Gamma_{uv}}(h) \qquad\text{ if }h_1 = h_2 = h \in \Z^{uv}\\
      0 ~~~~~~~~~~~~~~~~~~\qquad\text{ otherwise}\\
\end{array} 
\right. \]
There exists a polynomial time algorithm converting any collection of behavioral plans $(\sigma_1,\ldots,\sigma_n) \in \Sigma_1 \times \ldots \times \Sigma_n$ into a collection of vectors $(x_1,\ldots,x_n) \in \X_1 \times \ldots \times \X_n$ such that for any $u\in V$,
\[U_u(\sigma) = x_u ^\top \cdot \sum_{v:(u,v)\in E} A^{uv} \cdot x_v\]
In the opposite direction, 
there exists a polynomial time algorithm converting any collection of vectors $(x_1,\ldots,x_n) \in \X_1 \times \ldots \times \X_n$ into a collection of behavioral plans $(\sigma_1,\ldots,\sigma_n) \in \Sigma_1 \times \ldots \times \Sigma_n$  such that for any $u\in V$,
\[ x_u ^\top \cdot \sum_{v:(u,v)\in E} A^{uv} \cdot x_v = U_u(\sigma)\]
\end{lemma}

\begin{definition}\label{d:nash_seq}
A Nash Equilibrium of a network extensive form game $\G$ in sequence form is a vector $(x_1^\ast , \ldots, x_n^\ast) \in \X_1 \times \ldots \times \X_n$ such that for all $u \in V$: 
    \[(x^\ast_u) ^\top \cdot \sum_{v:(u,v)\in E} A^{uv} \cdot x^\ast_v \geq x_u^\top \cdot \sum_{v:(u,v)\in E} A^{uv} \cdot x^\ast_v ~~~\text{for all }x_u \in \X_u\]
\end{definition}
\begin{corollary}
Given a network extensive form game, any Nash Equilibrium $(\sigma^\ast_1,\ldots,\sigma^\ast_n) \in \Sigma_1 \times \ldots \times \Sigma_n$ (as per Definition~\ref{d:nash_beh}) can be converted in polynomial-time to a Nash Equilibrium $(x^\ast_1,\ldots,x^\ast_n) \in \X_1 \times \ldots \times \X_n$ (as per Definition~\ref{d:nash_seq}) and vice versa. 
\end{corollary}
The sequence form representation gives us a perspective with which we can analyze the theoretical properties of learning algorithms when applied to network zero-sum EFGs. In the following section, we utilize the sequence form representation to study a special case of Optimistic Mirror Descent known as Optimistic Gradient Ascent (OGA). 
\section{Our Convergence Results}
In this work, we additionally study the convergence properties of \textit{Optimistic Gradient Ascent} (OGA) when applied to \textit{network zero-sum EFGs}. OGA is a special case of \textit{Optimistic Mirror Descent} where the regularizer is $\psi(a) = \frac{1}{2}\|a\|^2$, which means that the Bregman divergence $D_{\psi}(x,y)$ equals $\frac{1}{2}\|x-y\|^2$. Since in network zero-sum EFGs each agent tries to maximize her payoff, OGA takes the following form:
\begin{align}
    x_u^t &= \text{argmax}_{x \in \X_u}\left\{ \eta \left\langle x, \sum_{v:(u,v)\in E} A^{uv} \cdot x_v^{t-1} \right\rangle - D_\psi\left(x,\hat{x}_u^t\right)\right\}~\label{eq:ODA1}\\
    \hat{x}_u^{t+1} &= \text{argmax}_{x \in \X_u}\left\{ \eta \left\langle x, \sum_{v:(u,v)\in E} A^{uv} \cdot x_v^t \right\rangle - D_\psi\left(x,\hat{x}_u^t\right)\right\}\label{eq:ODA2}
\end{align}
In Theorem~\ref{t:average} we describe the $\Theta(1/T)$ convergence rate to NE for the time-average strategies for any agent using OGA.  
\begin{theorem}\label{t:average}
Let $\{x^1,x^2, \ldots x^T\}$ be the vectors produced by Equations~(\ref{eq:ODA1}),(\ref{eq:ODA2}) for some initial strategies $x^0~:=~(x^0_1,\ldots,x^0_n)$. There exist game-dependent constants $c_1,c_2 >0$ such that if $\eta\leq 1/c_1$ then for any $u \in V$:
\[\hat{x}_u ^\top \cdot \sum_{v:(u,v)\in E} A^{uv} \cdot \hat{x}_v \geq x^\top \cdot \sum_{v:(u,v)\in E} A^{uv} \cdot \hat{x}_v - \Theta\left(\frac{c_1 \cdot c_2}{T}\right) ~~~\text{for all }x \in \X_u\]
where $\hat{x}_u = \sum_{s=1}^T x^s_u / T$.
\end{theorem}
Applying the polynomial-time transformation of Lemma~\ref{l:equiv_graph} to the time-average strategy vector $\hat{x} = (\hat{x}_1,\ldots,\hat{x}_n)$ produced by Optimistic Gradient Ascent, we immediately get that for any agent $u \in V$,
\[U_u(\hat{\sigma_u}, \hat{\sigma}_{-u}) \geq U_u(\sigma_u, \hat{\sigma}_{-u}) - \Theta\left(c_1\cdot c_2/T \right)~~~\text{for all }\sigma_u \in \Sigma_u\]
In Theorem~\ref{t:2} we establish the fact that OGA admits last-iterate convergence to NE in network zero-sum EFGs.

\begin{theorem}\label{t:2}
Let $\{x^1,x^2, \ldots x^T\}$ be the vectors produced by Equations~(\ref{eq:ODA1}),(\ref{eq:ODA2}) for $\eta\leq 1/c_3$ when applied to a network zero-sum extensive form game. Then, the following inequality holds:
\[\mathrm{dist}^2(x^t, \X^*) \leq 64\mathrm{dist}^2(x^1, \X^*)\cdot\left(1 + c_1\right)^{-t}\]
where $\X^\ast$ denotes the set of Nash Equilibria, $c_1 \coloneqq \min\left\{\frac{16\eta^2 c^2}{81}, \frac{1}{2} \right\}$ and $c_3, c$ are positive game-dependent constants.
\end{theorem}
We conclude the section by providing the key ideas towards proving Theorems~\ref{t:average} and~\ref{t:2}. For the rest of the section, we assume that the network extensive form game is consistent and zero-sum. Before proceeding, we introduce a few more necessary definitions and notations. We denote as $\X:= \X_1\times \ldots \times \X_n$ the product of treeplexes of Definition~\ref{d:treeplex_graph} and define the $|\X|\times |\X|$ matrix $R$ as follows:
\[ R_{(u:h_1),(v:h_2)} = \left\{
\begin{array}{ll}
      - \left[A^{uv}\right]_{h_1h_2} \qquad\text{ if }(u,v) \in E\\
      ~~~~~~~0 ~~~~~~~~~~~~\qquad\text{ otherwise}\\
\end{array} 
\right. \]
The matrix $R$ can be used to derive a more concrete form of the Equations~(\ref{eq:ODA1}),(\ref{eq:ODA2}):

\begin{lemma}\label{l:OGD_equivalence}
Let $\{x^1, x^2,\ldots,x^T\}$ be the collection of strategy vectors produced by Equations~(\ref{eq:ODA1}),(\ref{eq:ODA2}) initialized with $x^0:=(x_1^0,\ldots,x^0_n) \in \mathcal{X}$. The equations
\begin{align}
    x^t &= \text{argmin}_{x \in \X}\left\{ \eta \left\langle x, R\cdot x^{t-1} \right\rangle + D_\psi\left(x,\hat{x}^t\right)\right\}~\label{eq:OGD1}\\
    \hat{x}^{t+1} &= \text{argmin}_{x \in \X}\left\{ \eta \left\langle x, R \cdot x^t \right\rangle + D_\psi\left(x,\hat{x}^t\right)\right\}\label{eq:ODG2}
\end{align}
produce the exact same collection of strategy vectors $\{x^1,\ldots,x^T\}$ when initialized with $x^0 \in \mathcal{X}$.
\end{lemma}

To this end, we derive a \textit{two-player symmetric game $(R,R)$} defined over the polytope $\mathcal{X}$. More precisely, the $x$-agent selects $x \in \mathcal{X}$ so as to minimize $x^\top  R  y$ while the $y$-agent selects $y \in \mathcal{X}$ so as to minimize $y^\top R x$. Now consider the Optimistic Mirror Descent algorithm (described in Equations~(\ref{eq:Mirror_Descent1}),(\ref{eq:Mirror_Descent2})) applied to the above symmetric game. Notice that if $x^0 = y^0$, then by the symmetry of the game, the produced strategy vector $(x^t,y^t)$ will be of the form $(x^t,x^t)$ and indeed, $(x^t$, $\hat{x}^t)$ will satisfy Equations~(\ref{eq:OGD1}),~(\ref{eq:ODG2}). We prove that the produced vector sequence $\{x^t\}_{t \geq 1}$ converges to a \textit{symmetric Nash Equilibrium}.

\begin{lemma}\label{l:sne}
A strategy vector $x^\ast$ is an $\epsilon$-symmetric Nash Equilibrium for the symmetric game $(R,R)$ if the following holds:
\[(x^\ast)^\top \cdot R \cdot x^\ast \leq x^\top \cdot R \cdot x^\ast + \epsilon~~~~~\text{for all }x\in \mathcal{X} 
\]
Any $\epsilon$-symmetric Nash Equilibrium $x^\ast \in \mathcal{X}$ is also an $\epsilon$-Nash Equilibrium for the network zero-sum EFG.
\end{lemma}
A key property of the constructed matrix is the one stated and proven in Lemma~\ref{l:zero}. Its proof follows the steps of the proof of Lemma~B.3 in \cite{cai2011minmax} and is presented in Appendix~\ref{s:zero}.
\begin{lemma}\label{l:zero}
$x^\top \cdot R \cdot y + y^\top \cdot R \cdot x = 0$ for all $x,y \in \X$. 
\end{lemma}
Once Lemma~\ref{l:zero} is established, we can use to it to prove that the time-average strategy vector converges to an $\epsilon$-symmetric Nash Equilibrium in a two-player symmetric game.    
\begin{lemma}\label{l:average}
Let $(x^1, x^2,\ldots,x^T)$ be the sequence of strategy vectors produced by Equations~(\ref{eq:OGD1}),(\ref{eq:ODG2}) for $\eta~\leq~\mathrm{min}\{1/8\|R\|^2,1\}$. Then, 
\[\text{min}_{x \in \X} x^\top \cdot R \cdot \hat{x} \geq -\Theta\left(\frac{\mathcal{D}^2\|R\|^2}{T}\right) \]
where $\hat{x} = \sum_{s=1}^T x^s/T$ and $\mathcal{D}$ is the diameter of the treeplex polytope $\mathcal{X}$.
\end{lemma}
Combining Lemma~\ref{l:zero} with Lemma~\ref{l:average}, we get that that the time-average vector $\hat{x}$ is a $\Theta\left(\frac{\mathcal{D}^2\|R\|^2}{T}\right)$-symmetric Nash Equilibrium. This follows directly from the fact that $\hat{x}^\top \cdot R \cdot \hat{x} =0$. Then, Theorem~\ref{t:average} follows via a direct application of Lemma~\ref{l:sne}. For completeness, we present the complete proof of Theorem \ref{t:average} in Appendix~\ref{s:proof_theorem1}.

By Lemma~\ref{l:zero}, it directly follows that the set of symmetric Nash Equilibria can be written as: 
\[\X^\ast~=~\{x^\ast~\in~\X:~\min_{x~\in~\X} x^\top \cdot R \cdot x^\ast = 0\}.\] 
Using this, we further establish that Optimistic Gradient \emph{Descent} admits last-iterate convergence to the symmetric NE of the $(R,R)$ game. This result is formally stated and proven in Theorem~\ref{t:3}, the proof of which is adapted from the analysis of \cite{wei2021last}, with modifications to apply the steps to our setting. The proof of Theorem~\ref{t:3} is deferred to Appendix~\ref{s:proof_theorem2}. 
\begin{theorem}\label{t:3}
 Let $\{x^1,x^2, \ldots x^T\}$ be the vectors produced by Equations (\ref{eq:OGD1}),(\ref{eq:ODG2}) for $\eta\leq \min( 1/8\norm{R}^2,1)$. Then:
 \[\mathrm{dist}^2(x^t, \X^\ast) \leq 64\mathrm{dist}^2(x^1, \X^\ast)\cdot\left(1 + C_2\right)^{-t}\]
 where $C_2 \coloneqq \min\left\{\frac{16\eta^2 C^2}{81}, \frac{1}{2} \right\}$ with $C$ being a positive game-dependent constant.
 \end{theorem}
The statement of Theorem~\ref{t:2} then follows directly by combining Theorem~\ref{t:3} and Lemma~\ref{l:sne}. This result generalizes previous last-iterate convergence results for the setting of two-player zero-sum EFGs, even for games without a unique Nash Equilibrium.

\section{Experimental Results}\label{sec:experiments}
In order to better visualize our theoretical results, we experimentally evaluate OGA when applied to various network extensive form games. As part of the experimental process, for each simulation we ran a hyperparameter search to find the value of $\eta$ which gave the best convergence rate.\\

\textbf{Time-average Convergence.} Our theoretical results guarantee time-average convergence to the Nash Equilibrium set (Theorem \ref{t:average}). 
We experimentally confirm this by running OGA on a network version of the ubiquitous Matching Pennies game with 20 nodes (Figure \ref{fig:timeavg}~(a)), followed by a 4-node network zero-sum EFG (Figure~\ref{fig:timeavg}~(b)). In particular, for the latter experiment each bilinear game between the players on the nodes is a randomly generated extensive form game with payoff values in $[0,1]$. Next, we experimented with a well-studied simplification of poker known as Kuhn poker \cite{kuhn1950poker}. Emulating the illustrative example of a competitive online Poker lobby as described in Section \ref{sec:intro}, we modelled a situation whereby each agent is playing against multiple other agents, and ran simulations for such a game with $5$ agents (Figure \ref{fig:timeavg} (c)).

In the plots, we show on the $y$-axis the difference between the cumulative averages of the strategy probabilities and the Nash Equilibrium value calculated from the game. In each of the plots, we see that these time-average values go to $0$, implying convergence to the NE set.\\
\begin{figure}[!htb]
\centering
\begin{minipage}{.5\textwidth}
  \centering
  \includegraphics[width=.95\linewidth]{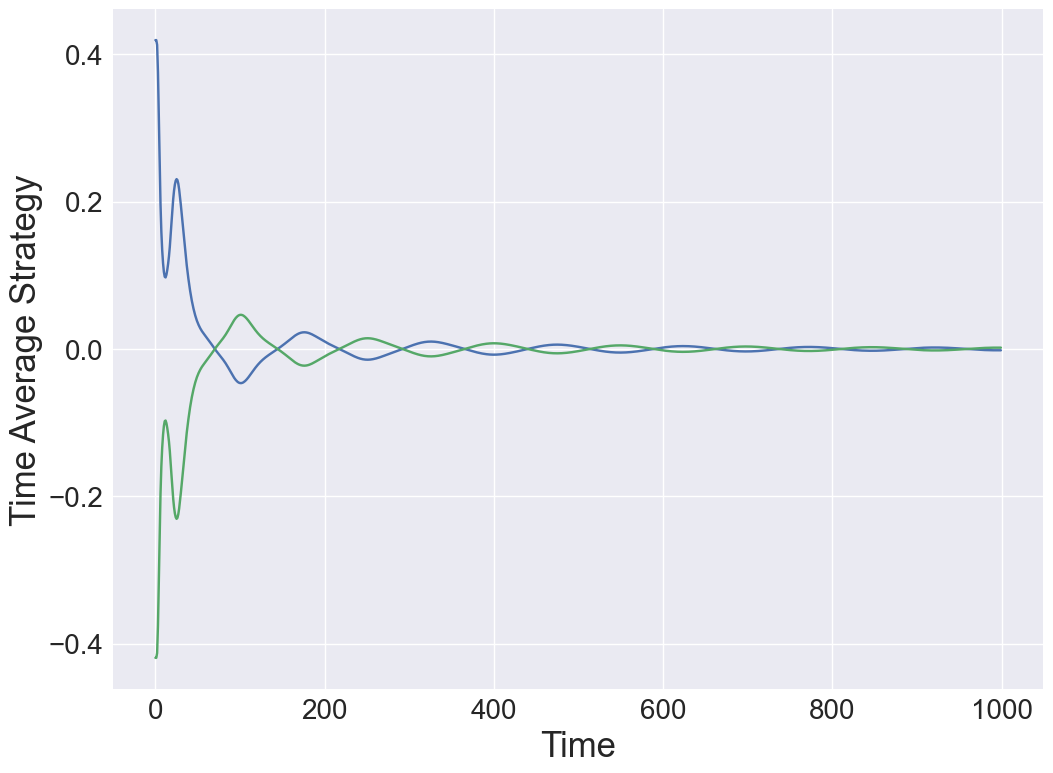}
\end{minipage}%
\hfill
\begin{minipage}{.5\textwidth}
  \centering
  \includegraphics[width=.95\linewidth]{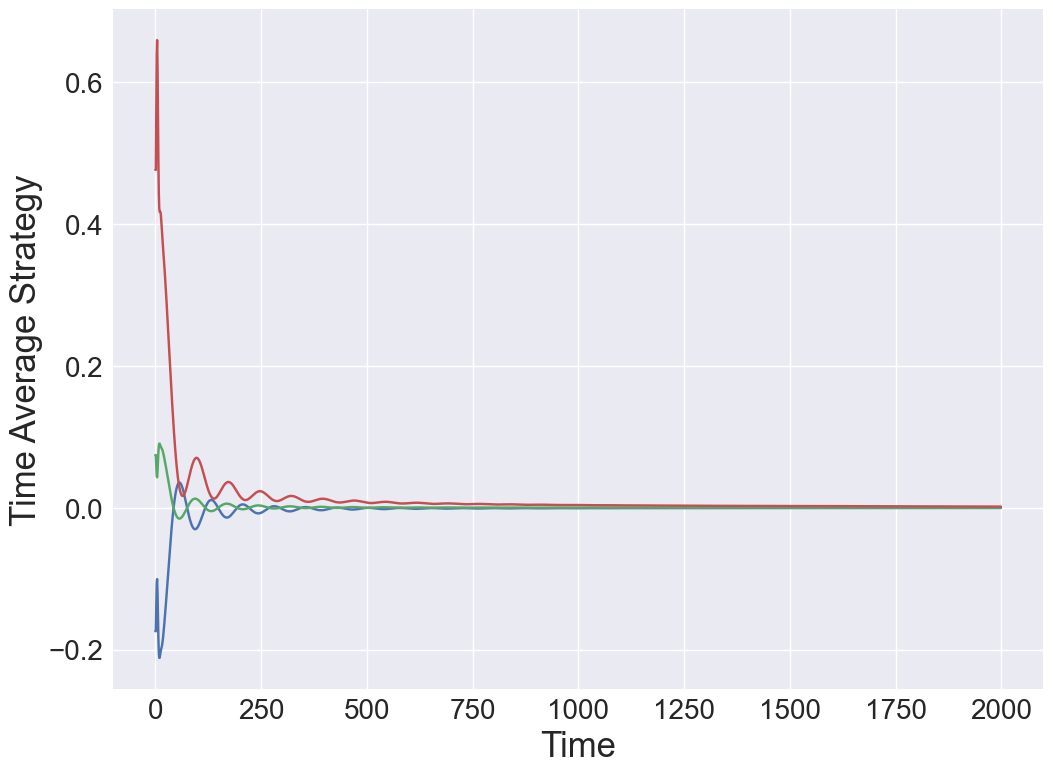}
\end{minipage}
\hfill
\begin{minipage}{.5\textwidth}
  \centering
  \includegraphics[width=.95\linewidth]{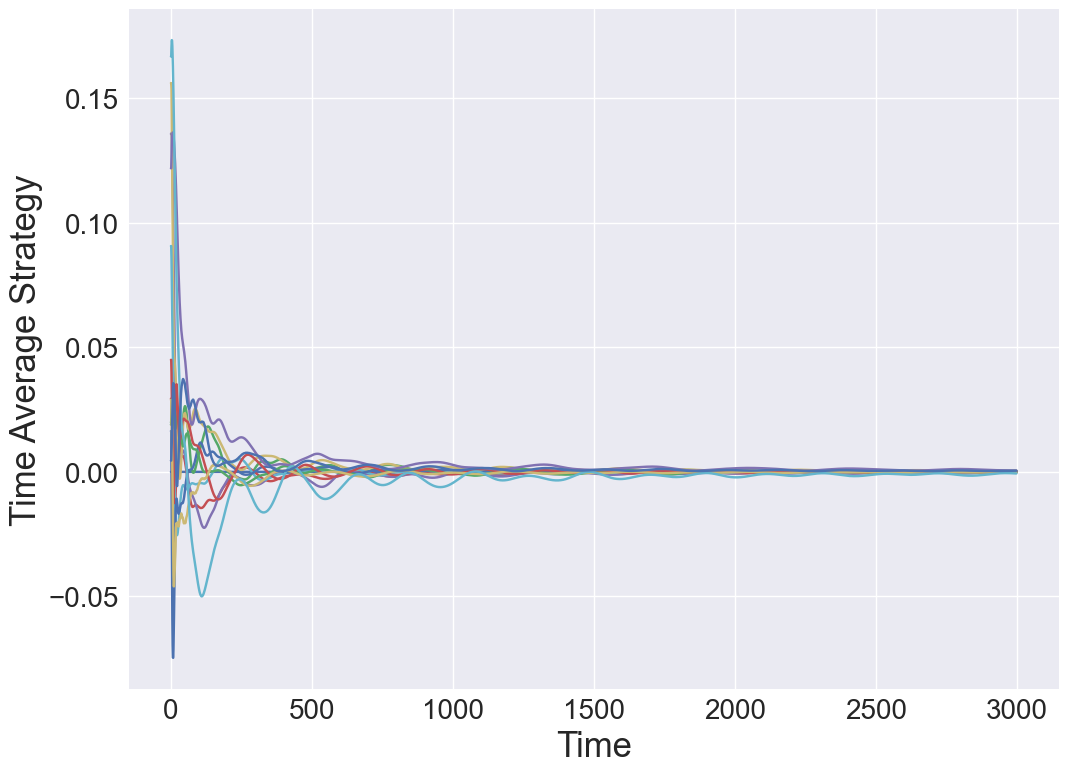}
\end{minipage}
\caption{Time-average convergence of OGA in network zero-sum extensive form games, where each player is involved in 2 or more different games and must select their strategy accordingly. (a) 20-node Matching Pennies game. (b) 4-node random extensive form game. (c) 5-node \emph{Kuhn poker} game.}
\label{fig:timeavg}
\end{figure}


\textbf{Last-iterate Convergence.}
Theorem \ref{t:2} guarantees $O(c^{-t})$ convergence in the last-iterate sense to a Nash Equlibrium for OGA. Similar to the time-average case, we ran simulations for randomly generated $3$ and $4$-node network extensive form games, where each bilinear game between two agents is a randomly generated matrix with values in $[0,1]$ (Figure \ref{fig:experiments} (a-b)). Moreover, we also simulated a 5-node game of Kuhn poker in order to generate Figure \ref{fig:experiments} (c). In order to generate the plots, we measured the log of the distance between each agent's strategy at time $t$ and the set of Nash Equilibria (computed \emph{a priori}), given by $\log(\dist^2(x^t,\mathcal{X}^*))$. As can be seen in Figure \ref{fig:experiments}, OGA indeed obtains fast convergence in the last-iterate sense to a Nash Equilibrium in each of our experiments.

A point worth noting is that when the number of nodes increases, the empirical last-iterate convergence time also increases drastically. For example, in the 5-player Kuhn poker game we see that each agents' convergence time is significantly greater compared to the smaller scale experiments. However, with a careful choice of $\eta$, we can still guarantee convergence to the set of Nash Equilibria for all players. Further discussion of these observations and detailed game descriptions can be found in Appendix \ref{appsec:experiments}.

\begin{figure}[!htb]
\centering
\begin{minipage}{.5\textwidth}
  \centering
  \includegraphics[width=.95\linewidth]{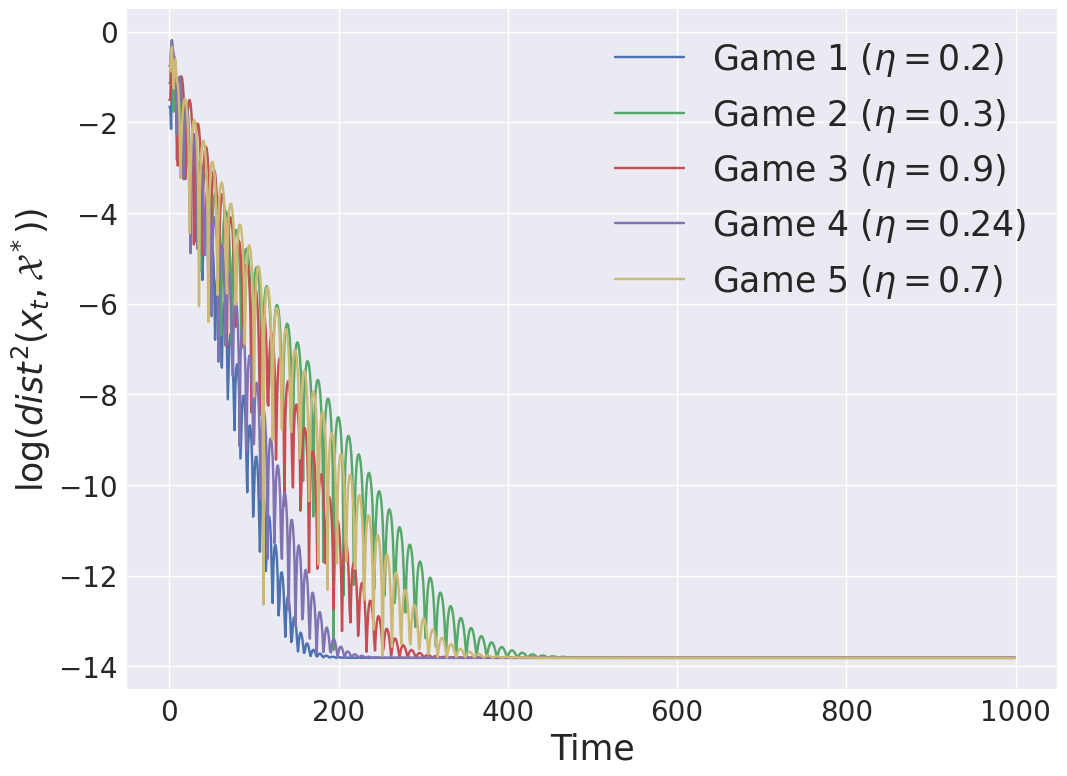}
\end{minipage}%
\hfill
\begin{minipage}{.5\textwidth}
  \centering
  \includegraphics[width=.95\linewidth]{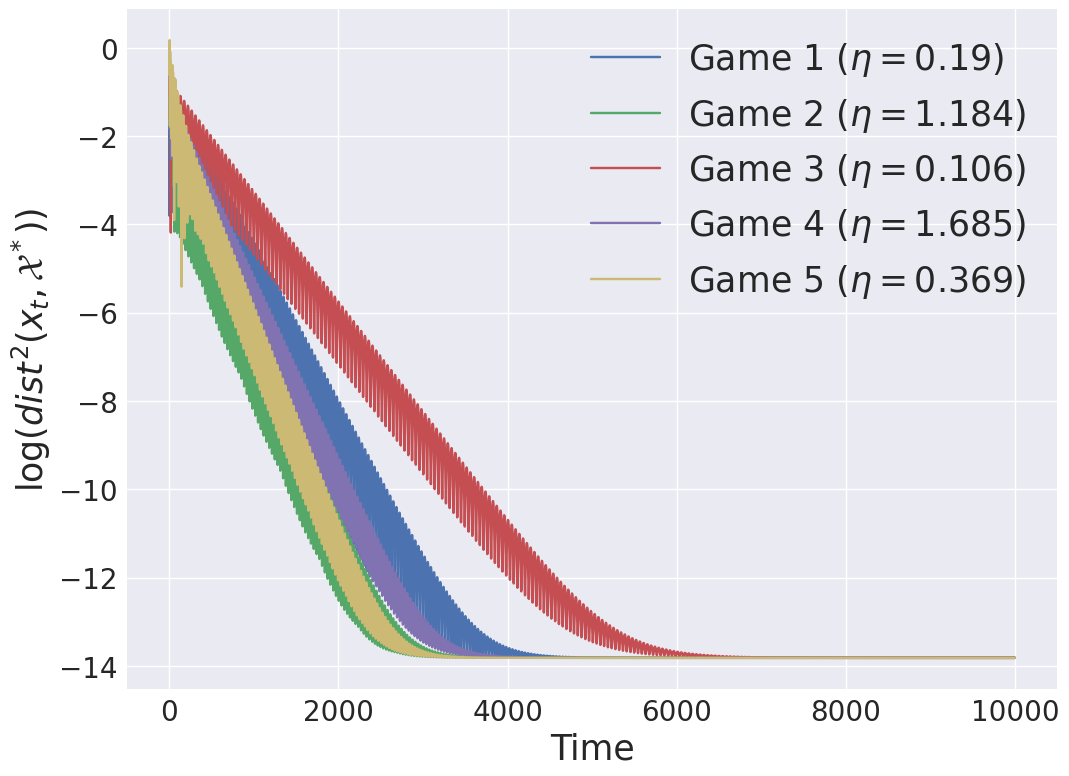}
\end{minipage}
\hfill
\begin{minipage}{.5\textwidth}
  \centering
  \includegraphics[width=.95\linewidth]{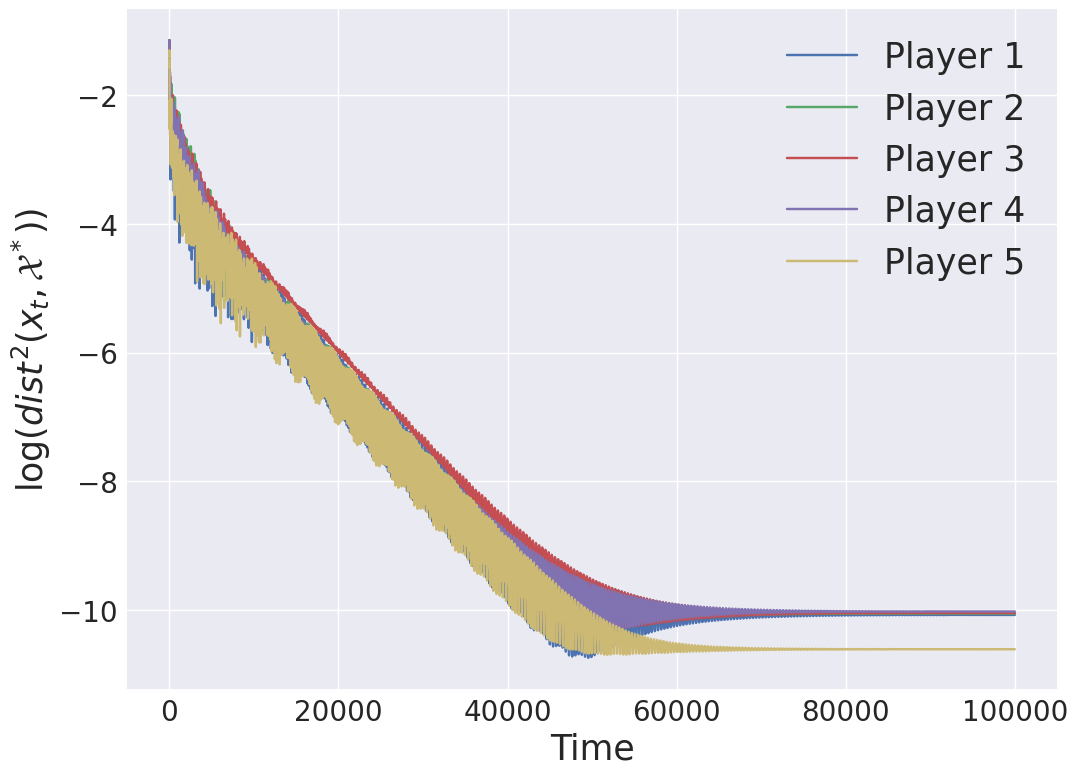}
\end{minipage}%
\caption{Last-iterate convergence of OGA to the NE in network zero-sum extensive form games. The plots shown are: (a) 3-node randomly generated network zero-sum extensive form game. (b) 4-node random network zero-sum extensive form game. Note the significantly longer time needed to achieve convergence compared to the 3-node experiment. (c) 5-node \emph{Kuhn poker} game. }
\label{fig:experiments}
\end{figure}

\section{Conclusion}
In this paper, we provide a formulation of \emph{Network Zero-Sum Extensive Form Games}, which encode the setting where multiple agents compete in pairwise games over a set of resources, defined on a graph. We analyze the convergence properties of \emph{Optimistic Gradient Ascent} in this setting, proving that OGA results in both time-average and day-to-day convergence to the set of Nash Equilibria. In order to show this, we utilize a transformation from network zero-sum extensive form games to two-player symmetric games and subsequently show the convergence results in the symmetric game setting. This work represents an initial foray into the world of online learning dynamics in network extensive form games, and we hope that this will lead to more research into the practical and theoretical applications of this class of games.


\clearpage
\subsubsection*{Acknowledgements}
This research/project is supported by the National Research Foundation Singapore and DSO National Laboratories under the AI Singapore Program (AISG Award No: AISG2-RP-2020-016), NRF2019-NRFANR095 ALIAS grant, grant PIE-SGP-AI-2020-01, NRF 2018 Fellowship NRF-NRFF2018-07 and AME Programmatic Fund (Grant No. A20H6b0151) from the Agency for Science, Technology and Research (A*STAR). Ryann Sim gratefully acknowledges support from the SUTD President's Graduate Fellowship (SUTD-PGF).

\bibliographystyle{abbrvnat}
\bibliography{ICLR_refs}
\newpage

\appendix
\section*{Appendix}
\section{Additional Related Work}
The related works presented in Section \ref{sec:intro} are primarily focused on research which is directly related to our topic of study, namely network generalizations of zero-sum extensive form games. However, there is a large body of work which studies many adjacent areas of interest. 

\paragraph{Extensive Form Games.} As elucidated in the main text, extensive form games are widely studied due to their numerous applications. The problem of computing Nash Equilibria in extensive form games is of major interest, with several works utilizing techniques such as CFR methods \cite{zinkevich2007regret} and LP methods \cite{shoham2008multiagent}.
Of particular note is the success of works utilizing CFR-based algorithms to study poker variants \cite{brown2018superhuman,bowling2015heads,moravvcik2017deepstack}.
Two-player EFGs can be written in sequence form (as described in the main text), which allows for them to be written as bilinear saddle-point problems. This connection allows for the design of algorithms that utilize first order methods to achieve approximate convergence to the Nash \cite{kroer2017smoothing,gao2019increasing}.


\paragraph{Online Learning in Games.} In this paper we study the properties of a particular online learning algorithm, Optimistic Gradient Ascent, for network zero-sum extensive form games. In normal form zero-sum games, recent results have shown that algorithms such as Gradient Descent Ascent and Multiplicative Weights Update do not converge in the last-iterate sense, even in the simplest of instances \cite{bailey2018multiplicative,vlatakis2019poincare}. In contrast, optimistic variants of these algorithms have been shown to be effective in guaranteeing last-iterate convergence \cite{daskalakis2017training,daskalakis2018last}. As described in the main text, some of these results have been extended to two-player extensive form games. Specifically, optimistic gradient descent and multiplicative weights update, as well as the versions thereof with \emph{dilated} regularizers, have been studied by \cite{lee2021last} and \cite{wei2020linear} in the two-player setting. This line of research into extensive form games is not limited to discrete time algorithms. \cite{perolat2021poincare} show that a continuous learning dynamic known as Follow the Regularized Leader (FTRL) exhibits last-iterate convergence in monotone two-player zero-sum EFGs.

\section{Omitted Proofs}\label{appsec:proofs}
\subsection{Proof of Lemma $1$}
We first describe how a behavioral plan $\sigma_i$ can be transformed to a vector $x_i(h)  \in \X_i$. For any $h \in \X_i$ we let $x_i(h):= \Pi_{(h,h')\in \P(h) \cap \X_i} \sigma_i(h ,\alpha_{h'})$ where $\alpha_{h'}$ is the action $\alpha \in \A(h)$ such that $h' = \mathrm{Next}(h,\alpha)$.
We set $x_i(h):= 1$ for all $h \in \H_i$ with $\mathrm{Prev}(h,i) = \varnothing$. Notice that by definition $U_1(\sigma) = \sum_{z \in \Z}x_1(z) \cdot p_1(z) \cdot x_2(z) = x_1^\top \cdot A_1^\Gamma \cdot x_2$ and respectively 
$U_2(\sigma) = \sum_{z \in \Z}x_2(z) \cdot p_2(z) \cdot x_1(z) = x_2^\top \cdot A_2^\Gamma \cdot x_1$.

Up next we show that all the constraints are satisfied. Consider the a state $h \in \H_i$ and the states $h' \in \mathrm{Next}(h,\alpha,i)$ for some $\alpha \in \A(h)$. Notice that for each $h'\in \mathrm{Next}(h,\alpha,i)$, $x_i(h') = x_i(h)\sigma_i(h , \alpha)$. This implies that $\sum_{\alpha \in \A(h)}x_i(\mathrm{Next}(h,\alpha,i)) = x_i(h)$ since $\sum_{\alpha \in \A(h)}\sigma_i(h,\alpha) = 1$.

Now let $h_1,h_2 \in \H_i$ where $h_1 \in \mathrm{Next}(h_1',\alpha,i)$, 
$h_2 \in \mathrm{Next}(h_2',\alpha,i)$ and $\I(h_1) = \I(h_2)$. Consider the set $\P(h_1) \cap \X_i := \{p_1,\ldots,p_k,h_1\}$ and $\P(h_2) \cap \X_i := \{q_1,\ldots,q_k,h_2\}$. Due to the perfect recall property, $m=k$ and $\I(p_\ell) = \I(q_\ell)$. Thus, $x_i(h_1) = x_i(h_2)$.

Up next we show how a vector $x_i \in \X_i$ can be converted to a behavioral plan $\sigma_i \in \Sigma_i$. Let $\sigma_i(h,\alpha):= \frac{x_i(h')}{x(h)}$ for some $h' \in \mathrm{Next}(h,\alpha,i)$. Notice that due the third constraint, $x_i(h') = x_i(h'')$ for all $h',h'' \in \mathrm{Next}(h,\alpha,i)$ and thus $\sigma(h,\alpha)$ is well-defined. For $h \in \H_i$ let $h_\alpha \in \mathrm{Next}(h,\alpha,i)$. By the third constraint we get that $\sum_{\alpha \in \A(h)}\sigma(h,\alpha) = 1$. Finally let $h_1,h_2 \in H_i$ with $\I(h_1) = \I(h_2)$ then $\sigma(h_1,\alpha)  = \frac{x_i(h'_1)}{x_i(h_1)}$ for some $h_1 \in \mathrm{Next}(h_1,\alpha,i)$ and $\sigma(h_2,\alpha)  = \frac{x_i(h'_2)}{x_i(h_2)}$ for some $h_2 \in \mathrm{Next}(h_2,\alpha,i)$.
As a result, by the second constraint we get that $\sigma(h_1,\alpha) = \sigma(h_2,\alpha)$ for all $\alpha \in \A(h)$.

\subsection{Proof of Lemma~\ref{l:equiv_graph}}
We first describe how a behavioral plan $\sigma_u \in \Sigma_u$ can be transformed to a vector $x_u(h)  \in \X_u$.
If there exists a game $\Gamma^{uv}$ with $(u,v) \in E$ such that $\mathrm{Prev}^{\Gamma^{uv}}(h,u) = \varnothing$ we set $x_u(h) \coloneqq 1$. Let us first verify that the above assignment is valid i.e. if $\mathrm{Prev}^{\Gamma^{uv}}(h,u) = \varnothing$ for some $(u,v)\in E$ then $\mathrm{Prev}^{\Gamma^{uv'}}(h,u) = \varnothing$ for all $(u,v') \in E$. Notice that $\P^{uv}(h)\cap \X_u = \{h\}$ and thus by the second constraint of Definition~\ref{d:treeplex_graph}, $\P^{uv'}(h)\cap \X_u = \{h\}$ for all $(u,v')\in E$. Now for the remaining nodes $h\in \H_u$ we select an arbitrary two-player EFG $\Gamma^{uv}$ ($(u,v) \in E$) containing the state $h$ and set $x_u(h):= \Pi_{(h,h')\in \P^{uv}(h) \cap \X_u} \sigma_u(h ,\alpha_{h'})$
where $\alpha_{h'}$ is the action $\alpha \in \A(h)$ such that $h' = \mathrm{Next}^{\Gamma^{uv}}(h,\alpha,u)$. We again need to argue that $x_u(h)$ is independent of the arbitrary choice of the game $\Gamma^{uv}$. Let assume that state $h$ also belongs in the two-player EFG $\Gamma^{uv'}$ for some $(u,v') \in E$. Again by the second constraint of Definition~\ref{d:consistency} we know that for the sets 
$\P^{uv}(h)\cap \X_u = \{p_1,\ldots,p_k,h\}$ and $\P^{uv'}(h)\cap \X_u = \{q_1,\ldots,q_m,h\}$ the following holds:
\begin{enumerate}
    \item $ k = m$.
    \item $\I(p_\ell) = \I(q_\ell)~~$ for all $\ell \in \{1,\ldots, k\}$.
    
    \item $p_{\ell + 1} \in \mathrm{Next}^{\Gamma_{uv}}(p_\ell,\alpha,u) $ and $q_{\ell + 1} \in \mathrm{Next}^{\Gamma_{uv'}}(q_\ell,\alpha,u) $ for some action $\alpha \in \A(p_\ell)$. 
\end{enumerate}
Since $\I(p_\ell) = \I(q_\ell)$ means that $\sigma_u(p_\ell,\alpha) = \sigma_u(q_\ell,\alpha)$ for all $\alpha \in \A(p_\ell)=\A(q_\ell)$, we get that \[\Pi_{(h,h')\in \P^{uv}(h) \cap \X_u} \sigma_u(h ,\alpha_{h'}) = \Pi_{(h,h')\in \P^{uv'}(h) \cap \X_u} \sigma_u(h ,\alpha_{h'})\]

Conversely, we show how a strategy in sequence form $x_u \in \X_u$ can be converted to behavioral plan $\sigma_u \in \Sigma_u$. Given a state $h \in \H_u$ we consider an edge $(u,v)\in E$ such that $\Gamma^{uv}$ containing $h \in \H_u$ and set 
\[ \sigma(h,\alpha) := \frac{x_u(h')}{x_u(h)}~~~\text{for some }h' \in \mathrm{Next}^{\Gamma^{uv}}(h,\alpha, u)\]
We first need to show that this is a valid probability distribution, $\sum_{h \in \A(h)}\sigma_u(h,\alpha) = 1$. Since $x_u \in \X_u^{\Gamma^{uv}}$, the second constraint of Definition~\ref{d:treeplex} ensures that 
\[\sum_{\alpha \in \A(h)}x_u (\mathrm{Next}(h,\alpha,u)) = x_u(h)\]
The latter implies that $\sum_{h \in \A(h)}\sigma_u(h,\alpha) = 1$.

We now need to establish that $\sigma(h,\cdot)$ is independent of the selection of the edge $(u,v) \in E$. Let $h$ be a state of the game $\Gamma^{uv'}$ for some $(u,v') \in E$. By constraint~$2$ of Definition~\ref{d:treeplex_graph}, for any $h' \in \mathrm{Next}^{\Gamma^{uv}}(h,\alpha,u)$ and $h'' \in \mathrm{Next}^{\Gamma^{uv'}}(h,\alpha,u)$ we have that $x_u(h') =x_u(h'')$ and thus $\sigma(h,\alpha) = \frac{x_u(h'')}{x_u(h)}$.

Finally we need to argue that if $h_1,h_2 \in \H_u$ with $\I(h_1) = \I(h_2)$, then $\sigma(h_1,\alpha) = \sigma(h_2,\alpha)$ for all $\alpha \in \A(h_1) = \A(h_2)$. Let $\sigma(h_1,\alpha)= \frac{x_u(h_1')}{x_u(h)}$ for some $h_1' \in \mathrm{Next}(h_1,\alpha,u)$ and $\sigma(h_1,\alpha)= \frac{x_u(h_1')}{x_u(h)}$ for some $h_2' \in \mathrm{Next}(h_2,\alpha,u)$. Then by Constraint~$3$ of Definition~\ref{d:consistency} we get that $x(h_1') = x(h_2')$ and thus $\sigma(h_1,\alpha) = \sigma(h_2,\alpha)$.

\subsection{Proof of Lemma \ref{l:OGD_equivalence}}
First, since Equations (\ref{eq:OGD1}), (\ref{eq:ODG2}) are defined on the product of treeplexes $\X$, let us decompose the equations from the perspective of an arbitrary agent $u$. Specifically, for some $x_u^t$, $u\in\{1,\dots,n\}$ it holds that the inner product $\langle x, R\cdot x^{t-1} \rangle$, $x\in \X$ can be decomposed into inner products of the form $\langle x, R\cdot x_u^{t-1} \rangle$, where $x$ is now in the individual treeplex $\X_u$. Moreover, by the definition of matrix $R$, we can substitute the following:
\[
R_{(u:h_1), (v:h_2)} = - [A^{uv}]_{h_1h_2}
\]
for all $(u,v)\in E$ and $0$ otherwise. Effectively, from the perspective of player $u$, the product of $R$ and $x^{t}$ gives us $\sum_{(u,v)\in E} A^{u,v}\cdot x_v^{t}$. This gives us the following:
\begin{align}
    x_u^t &= \text{argmin}_{x \in \X_u}\left\{ \eta \left\langle x, -\sum_{v:(u,v)\in E} A^{uv}\cdot x_v^{t-1} \right\rangle + D_\psi\left(x,\hat{x}_u^t\right)\right\}\\
    \hat{x}_u^{t+1} &= \text{argmin}_{x \in \X_u}\left\{ \eta \left\langle x,-\sum_{v:(u,v)\in E} A^{uv}\cdot x_v^{t} \right\rangle + D_\psi\left(x,\hat{x}_u^t\right)\right\}
\end{align}
Finally, we can just take the negative of the terms inside the braces to obtain Equations (\ref{eq:ODA1}), (\ref{eq:ODA2}). Hence, for every strategy vector $x$ updated using Equations (\ref{eq:OGD1}),(\ref{eq:ODG2}), the constituent strategy vectors for each player $u$ are exactly the same as Equations (\ref{eq:ODA1}), (\ref{eq:ODA2}). Thus if the initial conditions $x^0$ are the same, for all time $t$ the collection of strategy vectors $\{x^1,\dots,x^T\}$ are the same between both formulations.

\subsection{Proof of Lemma~\ref{l:sne}}
Let $\hat{x}:= (\hat{x}_1,\ldots,\hat{x}_n)$ be an $\epsilon$-symmetric Nash Equilibrium. Now consider the vector $x' \in \mathcal{X}$ defined as follows: $x_{u'} = \hat{x}_{u'}$ for all $u' \neq u$ and $x'_u$ is an arbitrary vector in $\X_u$. By the definition of the $\epsilon$-symmetric Nash Equilibrium we get that
\[\hat{x}^\top \cdot R \cdot \hat{x} - (x')^\top \cdot R \cdot \hat{x} \leq \epsilon\]
Notice that $(x')^\top \cdot R \cdot \hat{x} = -\sum_{v:(u,v)\in E}(x'_u)^\top \cdot A^{uv}\cdot \hat{x}_v - \sum_{u' \neq u}\sum_{v:(u',v)\in E}\hat{x}_{u'}^\top \cdot A^{u'v}\cdot \hat{x}_v$.
Thus we get
\[-\sum_{v:(u,v)\in E}(x'_u)^\top \cdot A^{uv}\cdot \hat{x}_v + \sum_{v:(u,v)\in E}(\hat{x}_u)^\top \cdot A^{uv}\cdot \hat{x}_v
\geq -\epsilon~~~~\text{for all }x_{u} \in \X_u
\]
Theorem~\ref{t:average} follows by repeating the same argument for all agents $u \in V$.

\subsection{Proof of Lemma~\ref{l:zero}}\label{s:zero}
We first prove a simpler version of Lemma~\ref{l:zero} where $x = y \in \mathcal{X}$. 
\begin{lemma}\label{l:zero_simple}
$x^\top \cdot R \cdot x = 0$ for all $x \in \X$.
\end{lemma}
\begin{proof}
Consider a vector $x \in \X$. To simplify notation let $x \coloneqq (x_1,\ldots,x_n)$ where each vector $x_u \in \X_u$. Let $\sigma_u^x \in \Sigma$ denote the behavioral plan for agent $u$ constructed from the vector $x_u \in \X_u$ as described in Lemma~\ref{l:equiv_graph}. By the zero-sum property of Definition~\ref{d:zero_sum}, we get that 
\[\sum_{u \in V}\sum_{v:(u,v)\in E} U^{uv}_u(\sigma^x_u,\sigma^x_v) = 0\]
By Lemma~\ref{l:equiv_graph} we get that $U^u(\sigma^x) = \sum_{v:(u,v)\in E}U_u^{uv}(\sigma^x_u,\sigma^x _v) = \sum_{v:(u,v)\in E} x_u^\top \cdot A^{uv}\cdot x_v$ meaning that 
\[\sum_{u \in V}\sum_{v:(u,v)\in E} x_u^\top \cdot A^{uv} \cdot x_v = 0\]
As a result, we get that $x^\top \cdot R \cdot x = 0$.
\end{proof}
We will also utilize the following result:
\begin{lemma}\label{l:prev}
Consider a node $u\in V$ and its neighbors $\N_u = \{v_1, v_2, \ldots, v_k\}$. Let $x_u \in \X_u$ represent a mixed strategy for
$u$ and $x_v$ a mixed strategy of the neighbor $v\in \N_u$. For any fixed collection $\{x_v\}_{v \in \N_u}$
the quantity 
\[ \sum_{v \in \N_u} x_v^\top \cdot A^{vu}\cdot x_u + 
\sum_{v \in \N_u} x_u^\top \cdot A^{uv}\cdot x_v
\]
remains constant over the range of $x_u$.
\end{lemma}
\begin{proof}
For any vector $x:= (x_1,\ldots,x_n) \in \X$, consider the vector $x' \in \X$ such that $x'_v = x_v$ for all $v \neq u$. By Lemma~\ref{l:zero_simple} we get that
\[x^\top \cdot R \cdot x - (x')^\top \cdot R \cdot x' = 0\]
The latter directly implies that
\[ \sum_{v \in \N_u} x_v^\top \cdot A^{vu}\cdot x_u + 
\sum_{v \in \N_u} x_u^\top \cdot A^{uv}\cdot x_v = \sum_{v \in \N_u} x^\top_v \cdot A^{vu}\cdot x'_u + 
\sum_{v \in \N_u} (x'_u)^\top \cdot A^{uv}\cdot x_v \]
for all $x_u,x'_u \in \X_u$.
\end{proof}

\begin{proof}[Proof of Lemma~\ref{l:zero}]
Consider vectors $x,y \in \X$. Consider the vector $y' \in \X$ such that $y_v = y'_v$ for all $v \neq u$. We first show that
\[x^\top \cdot R \cdot y + y^\top \cdot R \cdot x = x^\top \cdot R \cdot y' + (y')^\top \cdot R \cdot x\]
Let $\N_u$ denote the neighbors of agent $u \in V$,
\begin{eqnarray*}
&& x^\top \cdot R \cdot y + y^\top \cdot R \cdot x - x^\top \cdot R \cdot y' - (y')^\top \cdot R \cdot x \\
&=& \sum_{v \in \N_u} x_v^\top \cdot A^{vu} \cdot y_u + \sum_{v \in \N_u} x_u^\top \cdot A^{uv} \cdot y_v + \sum_{v \in \N_u} y_v^\top \cdot A^{vu} \cdot x_u + \sum_{v \in \N_u} y_u^\top \cdot A^{uv} \cdot x_v\\
&\quad\quad~&- \sum_{v \in \N_u} x_v^\top \cdot A^{vu} \cdot y'_u - \sum_{v \in \N_u} x_u^\top \cdot A^{uv} \cdot y'_v - \sum_{v \in \N_u} (y'_v)^\top \cdot A^{vu} \cdot x_u - \sum_{v \in \N_u} (y'_u)^\top \cdot A^{uv} \cdot x_v\\
&=& \sum_{v \in \N_u} x_v^\top \cdot A^{vu} \cdot y_u - \sum_{v \in \N_u} x_u^\top \cdot A^{uv} \cdot y_v - \sum_{v \in \N_u} x_v^\top \cdot A^{vu} \cdot y'_u - \sum_{v \in \N_u} (y'_u)^\top \cdot A^{uv} \cdot x_v\\
&=& 0
\end{eqnarray*}
where the last equality follows by Lemma~\ref{l:prev}. By gradually transforming vector $y$ to vector $x$ we get that $x^\top \cdot R \cdot y + y^\top \cdot R \cdot x = 2 \cdot x^\top \cdot R \cdot x = 0$.
\end{proof}

\subsection{Proof of Lemma~\ref{l:average}}

Applying Lemma~$1$ of \cite{rakhlin2013online} to our setting, we obtain:

\begin{lemma}[\cite{rakhlin2013online}]\label{lemma:R03}
Let $\{x^t,\hat{x}^t\}$ be the sequences produced by Equations~(\ref{eq:OGD1}),(\ref{eq:ODG2}). Then,
\begin{eqnarray*}
\sum_{t=1}^T (x^t)^\top \cdot R \cdot x^t - \min_{x \in \mathcal{X}}\sum_{t=1}^T x^\top \cdot R \cdot x^t &\leq &\frac{\mathcal{D}^2}{\eta} + \frac{1}{2}\sum_{t=1}^T \| R \cdot x^t - R\cdot x^{t-1} \|^2 \\
& &+ 
\frac{1}{2}\sum_{t=1}^T \| x^t - \hat{x}^t \|^2
-
\frac{1}{2\eta} \sum_{t=1}^T \left[ 
\|\hat{x}^t -x^t\|^2 + \|\hat{x}^{t} -x^{t+1}\|^2\right]
\end{eqnarray*}
where $\mathcal{D}$ is the diameter of the treeplex polytope $\mathcal{X}$.
\end{lemma}


Setting $\eta = \mathrm{min}\{1/(8\cdot\|R\|^2),1\}$ in Lemma~\ref{lemma:R03} we get that 
\begin{eqnarray*}
    && \sum_{t=1}^T (x^t)^\top \cdot R \cdot x^t - \min_{x \in \mathcal{X}}\sum_{t=1}^T x^{\top} \cdot R \cdot x^t \\
    &\leq& \frac{\mathcal{D}^2}{\eta} + \frac{1}{2}\sum_{t=1}^T \| R \cdot x^t - R \cdot x^{t-1} \|^2 -
    \frac{1}{4\eta} \cdot \sum_{t=1}^T \left[ 
    \|\hat{x}^t -x^t\|^2 + \|\hat{x}^{t} -x^{t+1}\|^2\right]\\
    &\leq& \frac{\mathcal{D}^2}{\eta} + \frac{1}{2}\sum_{t=1}^T \| R \cdot x^t - R \cdot x^{t-1} \|^2
    -
    2 \|R\|^2 \cdot \sum_{t=1}^T \left[ 
    \|\hat{x}^t -x^t\|^2 + \|\hat{x}^{t} -x^{t+1}\|^2\right]\\
    &\leq& \frac{\mathcal{D}^2}{\eta} + \frac{\|R\|^2}{2}\sum_{t=1}^T \|x^t -  x^{t-1} \|^2 - \norm{R}^2\cdot\sum_{t=1}^T \|x^t -  x^{t-1} \|^2 \\
    &\leq& \frac{\mathcal{D}^2}{\eta}
    \end{eqnarray*}
    Setting $\hat{x} = \sum_{s=1}^T x^s /T$ and using the fact that $(x^t)^\top \cdot R \cdot x^t = 0$ we get
    $\min_{x\in \mathcal{X}} x^\top \cdot R \cdot \hat{x} \geq -\frac{\mathcal{D}^2\|R\|^2}{T}$.

\subsection{Proof of Theorem~\ref{t:average}}\label{s:proof_theorem1}
Let $\hat{x}$ the time-average vector produced by Equations~(\ref{eq:OGD1}),(\ref{eq:ODG2}). By Lemma~\ref{l:average}, we have
\[\text{min}_{x \in \X} x^\top \cdot R \cdot \hat{x} \geq -\Theta\left(\frac{\mathcal{D}^2\|R\|^2}{T}\right) \]
Using the fact that $\hat{x}^\top \cdot R \cdot \hat{x} = 0$ we get that
\[\hat{x}^\top \cdot R \cdot \hat{x} \leq \text{min}_{x \in \X} x^\top \cdot R \cdot \hat{x} + \Theta\left(\frac{\mathcal{D}^2\|R\|^2}{T}\right) \]
meaning that $(\hat{x},\hat{x})$ is a $\Theta\left(\frac{\mathcal{D}^2\|R\|^2}{T}\right)$-approximate symmetric Nash Equilibrium of the symmetric game $(R,R)$. By Lemma~\ref{l:sne} we get that $\hat{x}$ is a $\Theta\left(\frac{\mathcal{D}^2\|R\|^2}{T}\right)$-approximate NE for the original network zero-sum EFG. 

\subsection{Proof of Theorem \ref{t:3}}\label{s:proof_theorem2}
First of all, in the proof of this theorem and in the lemmas presented within the proof, let $\X^\ast := \{x^\ast \in  \X:~\min_{x \in \X}x^\top \cdot R \cdot x^\ast = 0 \}$, which describes the set of symmetric Nash Equilibria.

In order to establish Theorem~\ref{t:3}, we follow the approach and notation of \cite{wei2020linear}, with minor modifications along the way to apply the steps to our setting. Applying Lemma~$1$ of \cite{wei2020linear}) to the Equations~(\ref{eq:OGD1}),~(\ref{eq:ODG2}) we get the following lemma:
\begin{lemma}[\cite{wei2020linear}]\label{l:rakhlin1}
    Let $\{x^t,\hat{x}^t\}_{t\geq 1}$ be the sequence of strategy vectors produced by Equations~(\ref{eq:OGD1}),~(\ref{eq:ODG2}) for $\eta \leq 1/8\norm{R}^2$. Then, 
    \[
    \eta (R\cdot x^t)^\top (x^t - x) \leq D_{\psi}(x, \hat{x}^t) - D_{\psi}(x, \hat{x}^{t+1}) - D_{\psi}(\hat{x}^{t+1}, x^t) - \frac{15}{16}D_{\psi}(x^t, \hat{x}^t) + \frac{1}{16}D_{\psi}(\hat{x}^t, x^{t-1})
    \]
\end{lemma}

Since for OGD we have that $D_\psi(x) = \frac{1}{2}\|x\|^2$, we can write the above inequality as:
\begin{equation}\label{eq:1}
    2\eta (R\cdot x^t)^\top (x^t - x) \leq \|\hat{x}^t - x\|^2 - \|\hat{x}^{t+1} - x\|^2 - \|\hat{x}^{t+1} - x^t\|^2 - \frac{15}{16}\|x^t - \hat{x}^t\|^2 + \frac{1}{16}\|\hat{x}^t - x^{t-1}\|^2
\end{equation}

To simplify notation let $x^\ast:= \Pi_{\X^*}(\hat{x}^t) \in \X^*$ meaning that $x^\ast$ is a symmetric Nash Equilibrium for the symmetric game $(R,R)$ and let us apply Equation~\ref{eq:1} with $x = x^\ast$. Now the LHS of Equation~\ref{eq:1} takes the following form
\begin{eqnarray}
2\eta (x^t)^\top \cdot R^T \cdot (x^t - x^\ast) &=&
-2\eta (x^t)^\top \cdot R^T \cdot x^\ast~~~~~~((x^t)^\top \cdot R^T \cdot x^t = 0) \nonumber\\
&=& -2\eta (x^\ast)^\top \cdot R \cdot x^t\nonumber \\
&=& -2\eta (x^t)^\top \cdot R \cdot x^\ast~~~~~~~~~\text{(by Lemma~\ref{l:zero})} \nonumber\\
&\geq& 0 \nonumber
\end{eqnarray}
where the last inequality follows by the fact that $(x^\ast,x^\ast)$ is a symmetric Nash Equilibrium of the game $(R,R)$. Since the LHS of Equation~\ref{eq:1} is greater or equal to $0$ we get that,
\[ 
    \|\hat{x}^{t+1} - \Pi_{\X^*}(\hat{x}^t)\|^2 \leq \|\hat{x}^t - \Pi_{\X^*}(\hat{x}^t)\|^2 - \|\hat{x}^{t+1} - x^t\|^2 - \frac{15}{16}\|x^t - \hat{x}^t\|^2 + \frac{1}{16}\|\hat{x}^t - x^{t-1}\|^2
\]
By definition, the left hand side of the above is bounded below by $\dist^2(\hat{x}^{t+1}, \X^*)$. Thus, we have the following inequality,
\begin{equation}\label{eq:2}
    \dist^2(\hat{x}^{t+1}, \X^*) \leq \dist^2(\hat{x}^{t}, \X^*) - \|\hat{x}^{t+1} - x^t\|^2 - \frac{15}{16}\|x^t - \hat{x}^t\|^2 + \frac{1}{16}\|\hat{x}^t - x^{t-1}\|^2  
\end{equation}
Now, we define $\Theta^t := \|\hat{x}^t - \Pi_{\X^*}(\hat{x}^t)\|^2 + \frac{1}{16}\|\hat{x}^t - x^{t-1}\|^2$ and $\xi^t := \|\hat{x}^{t+1} - x^t\|^2 + \|x^t - \hat{x}^t\|^2$ and rewrite Equation~\ref{eq:2} as follows,
\begin{equation}\label{eqn:theta}
    \Theta^{t+1} \leq \Theta^t - \frac{15}{16}\xi^t
\end{equation}

As in \cite{wei2020linear}, we now lower bound $\xi^t$ by a quantity related to $\dist^2(\hat{x}^{t+1}, \X^*)$ which will then give us a convergence rate for $\Theta^t$. To do so we need to establish a property that is known as saddle-point metric subregularity (\cite{wei2020linear}).
\begin{lemma}\label{l:metric}(Saddle-Point Metric Subregularity (SP-MS)) For any $x,x'\in \X\setminus \X^*$,
\[
\sup_{x'\in \X} \frac{(R\cdot x)^\top (x-x')}{\| x-x'\|} \geq c \cdot \|x -  \Pi_{\X^*} (x) \|
\]
for some game-dependent constant $c > 0$.
\end{lemma}
We present the proof of Lemma~\ref{l:metric} in Section~\ref{s:metric}. To this end, we remark that once the proof of Lemma~\ref{l:metric} is established, the proof
of Theorem~\ref{t:3} follows by the analysis of \cite{wei2020linear}. For the sake of completeness, we conclude the section with this analysis.

\begin{lemma}[\cite{wei2020linear}]\label{lemma:W04}
    If the parameter $\eta$ in Equations~(\ref{eq:OGD1}),~(\ref{eq:ODG2}) is selected less than $1/8\norm{R}^2$
then for any $t\geq 0$ and $x' \in \X$ with $x' \neq \hat{x}^{t+1}$,
    \[
        \|\hat{x}^{t+1} - x^t\|^2 + \|x^t - \hat{x}^{t}\|^2 \geq \frac{32}{81} \eta^2 \frac{\left[(R\cdot\hat{x}^{t+1})^\top(\hat{x}^{t+1}-x') \right]^2_+}{\|\hat{x}^{t+1} - x'\|^2}
    \]
    where $[a]_+ \coloneqq \max\{a, 0\}$, and similarly, for $x' \neq x^{t+1}$,
    \[
        \|\hat{x}^{t+1} - x^{t+1}\|^2 + \|x^t - \hat{x}^{t+1}\|^2 \geq \frac{32}{81} \eta^2 \frac{\left[(R\cdot x^{t+1})^\top(x^{t+1}-x') \right]^2_+}{\|x^{t+1} - x'\|^2}
    \]
    
\end{lemma}

Now taking the telescoping sum of Equation \ref{eqn:theta} over $t$, we get:
\[
    \Theta^1 \geq \Theta^1 - \Theta^T \geq \frac{15}{16}\sum_{t=1}^{T-1} \xi^t \geq \frac{15}{16}\sum_{t=1}^{T-1}  \left(\|\hat{x}^{t+1} - x^t\|^2 + \|x^t - \hat{x}^t\|^2\right) \geq \frac{15}{32}\sum_{t=2}^{T-1} \| x^t - x^{t-1}\|^2
\]
where the final inequality follows due to strong convexity of $\frac{1}{2}\|x\|^2$. Now, since the rightmost term is a summation of nonnegative terms and is  upper bounded by a finite constant, we have that $\|x^{t-1} - x^t\| \to 0$ as $T \to \infty$. Thus, $x^t$ converges to a point as $T\to\infty$. In addition, due to Theorem \ref{t:average}, we know that the time-average value of the iterates converge to a Nash Equilibrium. Combining these two observations, we can thus conclude that $x^t$ indeed converges to a Nash Equilibrium in the last-iterate sense.

To show the explicit rate of convergence, we will require a few additional observations. First, note that the following inequality holds for Equation \ref{eqn:theta}:
\begin{equation}\label{eqn:inequality4}
    \|\hat{x}^{t+1} - x^t\|^2 \leq \xi^t \leq \frac{16}{15} \Theta^t \leq \dots\leq \frac{16}{15}\Theta^1
\end{equation}

Then we have:
\begin{align*}
    \xi^t &\geq  \frac{1}{2} \|\hat{x}^{t+1} - x^t\|^2 + \frac{1}{2}(\|\hat{x}^{t+1} - x^t\|^2 + \|x^t - \hat{x}^{t}\|^2)\\
    &\geq \frac{1}{2} \|\hat{x}^{t+1} - x^t\|^2 + \frac{16\eta^2}{81}\sup_{x'\in \X} \frac{\left[(R\cdot\hat{x}^{t+1})^\top(\hat{x}^{t+1}-x') \right]^2_+}{\|\hat{x}^{t+1} - x'\|^2} && \text{(Lemma \ref{lemma:W04})}\\
    &\geq \frac{1}{2} \|\hat{x}^{t+1} - x^t\|^2 + \frac{16\eta^2 C^2}{81} \| \hat{x}^{t+1} - \Pi_{\X^*} (\hat{x}^{t+1})\|^2 && \text{(SP-MS condition)}\\
    &\geq \min\left\{\frac{16\eta^2 C^2}{81}, \frac{1}{2} \right\} \left( \|\hat{x}^{t+1} - x^t\|^2 + \|\hat{x}^{t+1} -  \Pi_{\X^*} (\hat{x}^{t+1})\|^2 \right) && \text{(Equation \ref{eqn:inequality4})}\\
    &= C_2 \Theta^{t+1}
\end{align*}
Now, we can show the explicit convergence rate as follows. Combining the above inequality with Equation \ref{eqn:theta}, we obtain
\begin{equation}
    \Theta^{t+1} \leq \Theta^t - C_2 \Theta^{t+1}
\end{equation}
This immediately implies that $\Theta^{t+1} \leq (1+C_2)^{-1}\Theta^t$. By iteratively expanding the right hand side of the inequality, we can equivalently write: \begin{equation}\label{eqn:thm3}
    \Theta^{t} \leq (1+C_2)^{-t+1}\Theta^1 \leq 2 \Theta^1 (1+C_2)^{-t}
\end{equation}

Next, notice that $\Theta^1$ is precisely $\dist^2(\hat{x}^1, \X^*)$. Moreover, by using the triangle inequality, we can write:
\begin{align*}
    \dist^2(x^t, \X^*) &\leq \|x^t - \Pi_{\X^*}(\hat{x}^{t+1})\|^2\\
    &\leq 2\|\hat{x}^{t+1} - \Pi_{\X^*}(\hat{x}^{t+1})\|^2 + 2\|\hat{x}^{t+1} - x^t\|^2\\
    &\leq 32 \Theta^{t+1} \leq 32 \Theta^{t}
\end{align*}
    
Combining this observation with Equation \ref{eqn:thm3} we get that 
\[
    \dist^2(x^t, \X^*) \leq 64 \dist^2(x^1, \X^*) (1+C_2)^{-t}
\]
where $C_2 = \min\left\{\frac{16\eta^2 C^2}{81}, \frac{1}{2} \right\}$, which completes the proof of Theorem \ref{t:3}.

\subsection{Proof of Lemma~\ref{l:metric}}\label{s:metric}
Lemma~\ref{l:metric} follows easily from Lemma~\ref{l:10}, the proof of which is presented in Section~\ref{appsec:13}.

\begin{lemma}\label{l:10}
For any $x \in \mathcal{X}$ the following holds:
    \begin{equation}
    -\min_{x' \in \mathcal{X}} x'^\top R \cdot x \geq c \cdot \|x - \Pi_{\X^*}(x)\|.
    \end{equation}
for some game-dependent constant $c \in (0,1)$.
    \end{lemma}
\begin{proof}[Proof of Lemma \ref{l:metric}.]
    Consider the LHS of the inequality in Lemma~\ref{l:10} and note that $-\min_{x' \in \mathcal{X}} x'^\top R \Bar{x} = 0$ if and only if $\Bar{x} \in \mathcal{X}^\ast$.
    
    Let $\mathcal{D}$ denote the diameter of $\mathcal{X}$ which is assumed to be finite. Then,
    \begin{align*}
        \max_{x' \in \X} \frac{(R\cdot x)^\top (x-x')}{\|x-x'\|} &\geq \max_{x' \in \X}  \frac{1}{\mathcal{D}} (R\cdot x)^\top (x-x')\\
        &= \frac{1}{\mathcal{D}} \max_{x' \in \X} x^\top R^\top (x-x')\\
        &= \frac{1}{\mathcal{D}} \max_{x' \in \X} [x^\top R^\top x - x^\top R^\top x']\\
        &= \frac{1}{\mathcal{D}} \max_{x' \in \X} [- x^\top R^\top x'] && (x^\top R^\top x = 0)\\
        &= - \frac{1}{\mathcal{D}} \min_{x' \in \X} x^\top R^\top x'\\
        &= - \frac{1}{\mathcal{D}} \min_{x' \in \X} x'^\top R x\\
        &\geq \frac{c}{\mathcal{D}} \|x - \Pi_{\X^*} (x)\| && (\text{Lemma \ref{l:10}})
    \end{align*}
 \end{proof}
 
\subsection{Proof of Lemma~\ref{l:10}}\label{appsec:13}
The proof of this lemma follows the basic steps in the proof of Theorem~$5$ in \cite{wei2020linear}, with some necessary modifications. We remind the reader that for the purposes of the proof, we defined the set of symmetric Nash Equilibria as $\mathcal{X}^\ast=\{x^\ast: \text{min}_{x \in \mathcal{X}}x^\top \cdot R \cdot x^\ast = 0\}$. The proof is split into several auxiliary lemmas/claims, which can then be combined to show the required result.
\begin{lemma}\label{lemma:polytope}
The set $\X^\ast$ is a polytope. 
\end{lemma}
\begin{proof}
Let $x^\ast \in \X^\ast$ then $\min_{x \in \X} x^\top \cdot R \cdot x^\ast =0$. Since $\mathcal{X}$ is a polytope the minimum value is attained in one of the vertices of polytope $\X$, the set of which is denoted by $\V(\X)$. Thus 
\[\min_{x \in \X} x^\top \cdot R \cdot x^\ast = \min_{v \in \V(\X)} v^\top \cdot R \cdot x^\ast = 0\]
As a result, the set $\mathcal{X}^\ast$ can be equivalently described as the set of vector $x^\ast \in \X$ that additionally satisfy $v_i^\top \cdot R \cdot x^\ast \geq 0$ for all vertices $v_i \in \V(\X)$.
\end{proof}
Let us describe $\X$ in the following polytopal form:
\[\X:= \{x:~~ \alpha_i^\top \cdot x \leq \beta_i \text{ for } i = 1, \ldots L\}\]
where $L$ is a positive integer. Consider also the following polytopal form of the set $\X^\ast$ as
\[\X^*:= \{x^\ast \in \X:~~ c_i^\top \cdot x^\ast \geq 0 \text{ for } i = 1, \ldots, K\}\]
where $c_i : = v_i^\top \cdot R$ with $v_i$ denoting the $i$-th vertex of polytope $\X$ and $K$ denotes the number of different vertices.

Now fix a specific $x \in \X \setminus \X^\ast$ and let $x^\ast := \Pi_{\X^\ast}(x)$. The vector $x^\ast$ satisfies some of the polytopal constraints with equality. These constraints are called \emph{tight}, and without loss of generality we can assume that
\begin{itemize}
    \item $\alpha_i^\top \cdot x^\ast = \beta_i$ for $i=1,\ldots,\ell$
    
    \item $c_i^\top \cdot x^\ast = 0$ for $i=1,\ldots,k$
\end{itemize}

\begin{lemma}\label{lemma:tight}
    The vector $x\in \X$ violates at least one tight constraint of the form $\{c_i^\top \cdot x = 0 \text{ for }i=1,\ldots,k\}$.
\end{lemma}
\begin{proof}
Let assume that $\{c_i^\top \cdot x = 0 \text{ for }i=1,\ldots,k\}$. Since $x \notin \X^\ast$ there exists
at least one vertex $v \in \V(\X)$ such that $v^\top \cdot R \cdot x <0$. The latter implies that there exits $x' \in \X$ lying in line segment between $x$ and $x^\ast$ such that $v^\top \cdot R \cdot x \geq 0$ for all vertices $v \in \V(\X)$. The latter implies that $x' \in \X^\ast$ which contradicts with the fact that $x^\ast = \Pi_{\X^\ast}(x)$. 
\end{proof}
Now, note that the normal cone of $\X^\ast$ at $x^\ast$ is
\[\mathcal{N}_{x^\ast} = \left\{x'-x^\ast:~~x^\ast=\Pi_{\X^\ast}(x')\right\}\]
From a standard result in linear programming literature \cite{wei2020linear}, we know that the normal cone can be written in the following form:
\[\mathcal{N}_{x^\ast} = \left\{ \sum_{i=1}^\ell p_i \cdot \alpha_i + \sum_{i=1}^k q_i \cdot c_i \text{ for some }p_i,q_i \geq 0
\right\}\]
Again, following the steps of \cite{wei2020linear}, we have the following claim:
\begin{claim}\label{claim:1}
For any $x \in \X$ such that $x^\ast = \Pi_{x \in \X^\ast}(x)$ the vector $x - x^\ast$ belongs in the set
\[\mathcal{M}_{x^\ast} = \left\{ \sum_{i=1}^\ell p_i \cdot \alpha_i + \sum_{i=1}^k q_i \cdot c_i:~~~p_i,q_i \geq 0 , \alpha_j^\top \cdot \left( \sum_{i=1}^\ell p_i \cdot \alpha_i + \sum_{i=1}^k q_i \cdot c_i\right) \leq 0
\right\}\]
\end{claim}
\begin{proof}
    As mentioned previously, we know that $x-x^*$ belongs in the normal cone of $x^*$, $\mathcal{N}_{x^*}$. Thus it can be expressed as $\sum_{i=1}^\ell p_i \cdot \alpha_i + \sum_{i=1}^k q_i \cdot c_i$ with $p_i, q_i \geq 0$. As such, we need only additionally show that $x-x^*$ satisfies the following:
    \[
        \alpha_i^\top (x-x^*) \leq 0,~~~~ \forall i\in 1,\dots,\ell
    \]
    Notice that for all $i=1,\dots,\ell$, we have:
    \begin{align*}
        \alpha_i^\top (x-x^*) &= (\alpha_i^\top x^* - b_i) + \alpha_i^\top (x-x^*) && (\text{$i$-th constraint is tight at $x^*$})\\
        &= \alpha_i^\top (x^*+x-x^*) - b_i\\
        &= \alpha_i^\top x - b_i \leq 0 && (x\in \mathcal{X})
    \end{align*}
\end{proof}
\begin{claim}\label{claim:2}
    $x-x^*$ can be written as $\sum_{i=1}^\ell p_i \cdot \alpha_i + \sum_{i=1}^k q_i \cdot c_i$ with $0 \leq p_i, q_i \leq C'\|x-x^*\|$ for all $i$ and some problem-dependent constant $C' < \infty$.
\end{claim}
\begin{proof}
    Note that $\frac{x-x^*}{\|x-x^*\|} \in \mathcal{M}_{x^*}$ because $0 \neq x-x^* \in \mathcal{M}_{x^*}$ and $\mathcal{M}_{x^*}$ is a cone. Furthermore, $\frac{x-x^*}{\|x-x^*\|} \in \{v\in \mathbb{R}^M:\|v\|_{\infty}\leq 1\}$. Thus,  $\frac{x-x^*}{\|x-x^*\|} \in \mathcal{M}_{x^*} \cap \{v\in \mathbb{R}^M:\|v\|_{\infty}\leq 1\}$, which is a bounded subset of the cone $\mathcal{M}_{x^*}$.
    
    We will argue that there exists large enough $C' > 0$ such that:
    \[
    \left\{\sum_{i=1}^\ell p_i \cdot \alpha_i + \sum_{i=1}^k q_i \cdot c_i : 0 \leq p_i, q_i \leq C' , \forall i\right\}
     \supseteq \mathcal{M}_{x^*} \cap \{v\in \mathbb{R}^M:\|v\|_{\infty}\leq 1\} \coloneqq \mathcal{P}.
    \]
    
    First note that $\mathcal{P}$ is a polytope. For every vertex $\hat{v}$ of $\mathcal{P}$, the smallest $C'$ such that $\hat{v}$ belongs to the left-hand side set above is the solution to the following linear program:
    
    \begin{align*}
        \min_{p_i, q_i, C_{\hat{v}}'} \quad
        & C_{\hat{v}}' \\
        \text{s.t.} \quad
        & \hat{v} = \sum_{i=1}^\ell p_i \cdot \alpha_i + \sum_{i=1}^k q_i \cdot c_i, \quad 0 \leq p_i, q_i \leq C_{\hat{v}}'.
    \end{align*}
    Since $\hat{v} \in \mathcal{M}_{x^*}$, this LP is always feasible and admits a finite solution $C_{\hat{v}}' < \infty$. Now, let $C' = \max_{\hat{v}\in\mathcal{V}(\mathcal{P})}$ where $\mathcal{V}(\mathcal{P})$ is the set of all vertices of $\mathcal{P}$. Then, since any $v\in \mathcal{P}$ can be expressed as a convex combination of points in $\mathcal{V}(\mathcal{P})$, $v$ can thus be expressed as  $\sum_{i=1}^\ell p_i \cdot \alpha_i + \sum_{i=1}^k q_i \cdot c_i$ where $0 \leq p_i, q_i \leq C'$. As a result, $\frac{x-x^*}{\|x-x^*\|}$ can be written as $\sum_{i=1}^\ell p_i \cdot \alpha_i + \sum_{i=1}^k q_i \cdot c_i$ where $0 \leq p_i, q_i \leq C'$, so it follows that $x-x^*$ can be written as: $\sum_{i=1}^\ell p_i \cdot \alpha_i + \sum_{i=1}^k q_i \cdot c_i$ where $0 \leq p_i, q_i \leq C'\|x-x^*\|$.
    
\end{proof}

Now, again following \cite{wei2020linear}, we can piece together all of the auxiliary results to show Lemma \ref{l:10}. Define $A_i \coloneqq \alpha_i^\top (x-x^*)$ and $C_i \coloneqq c_i^\top (x-x^*)$. By Claim \ref{claim:2}, we can write $x-x^*$ as $\sum_{i=1}^\ell p_i \cdot \alpha_i + \sum_{i=1}^k q_i \cdot c_i$ where $0 \leq p_i, q_i \leq C'\|x-x^*\|$. Thus:
\[
    \sum_{i=1}^\ell p_i \cdot A_i + \sum_{i=1}^k q_i \cdot C_i = \left(\sum_{i=1}^\ell p_i \cdot \alpha_i + \sum_{i=1}^k q_i \cdot c_i\right)^\top (x-x^*) = \|x-x^*\|^2
\]
Moreover, since $x-x^* \in \mathcal{M}_{x^*}$ by Claim \ref{claim:1}, we have
\[
    \sum_{i=1}^\ell p_i \cdot A_i = \sum_{i=1}^\ell p_i \cdot \alpha_i \leq 0 
\]
and 
\[
    \sum_{i=1}^k q_i \cdot C_i \leq \left(\max_{i\in\{1,\dots, k\}}C_i\right) \sum_{i=1}^k q_i \leq \left(\max_{i\in\{1,\dots, k\}}C_i\right) kC'\|x-x^*\|
\]
The first inequality follows because $p_i \geq 0$. The second inequality follows because $\max_{i\in\{1,\dots, k\}}C_i > 0$ (by Lemma~\ref{lemma:tight}) and $0 \leq q_i \leq C'\|x-x^*\|$.
 
Combining the above, we obtain:
\[
    \max_{i\in\{1,\dots, k\}}C_i \geq \frac{1}{kC'} \|x-x^*\|
\]
Now, note that: 
\[
\max_{i\in\{1,\dots, k\}}C_i = \max_{i\in\{1,\dots, k\}} (c_i^\top x - d_i) \leq \max_{i\in\{1,\dots, |\mathcal{V}(\mathcal{X})|\}} (c_i^\top x - d_i) = \max_{x' \in \mathcal{X}} (x'^\top R x) 
\]
where the last equality follows from the formulation of problem constraints in the proof of Lemma \ref{lemma:polytope}. Finally, by combining the last two statements, we can conclude that
\[
    -\min_{x'\in\mathcal{X}} (x'^\top R x) \geq \frac{1}{kC'}\|x-x^*\|.
\]
Here $k$ and $C'$ only depend on the set of tight constraints at $x^*$. There are only finitely many sets of tight constraints, so there exists a constant $C>0$ such that $-\min_{x'\in\mathcal{X}} (x'^\top R x) \geq \frac{1}{kC'}\|x-x^*\|$ holds for all $x$ and $x^*$, completing the proof.

\section{Additional Experimental Details}\label{appsec:experiments}
In this section we provide more details about our experimental results from Section \ref{sec:experiments}. 

\textbf{Random Network Extensive Form Games.} In our simulations, we first generated random zero-sum extensive form games on both a $3$-node graph where every player plays against the other two players, as well as a dense $4$-node graph (shown in Figure \ref{fig:graph1}). Specifically, each game is characterized by a $3\times 3$ symmetric matrix which represents the sequence form of an extensive form game written as a matrix. For each run of the simulation, we first create the games which are to be played, randomly generating matrices with elements in $[0,1]$. Then, we optimize for the choice of stepsize $\eta$, selecting the value that gives the fastest convergence rate to the Nash Equilibrium. In the plots, in order to reduce visual clutter, we present the squared distance from the Nash for only one of the players. In addition, in order to more clearly show the fast rate of convergence, we compute the logarithm of $\dist^2(x^t, \X^*)$ in the plots. It is worth noting that the $4$-node graph takes significantly longer to arrive at the last iterate compared to the $3$-node graph. 
\begin{figure}
    \centering
    \begin{tikzpicture}
    \draw[fill=black] (0,0) circle (3pt);
    \draw[fill=black] (3,0) circle (3pt);
    \draw[fill=black] (0,3) circle (3pt);
    \draw[fill=black] (3,3) circle (3pt);
    \node at (-0.3,0) {4};
    \node at (3.3,0) {3};
    \node at (-0.3,3) {1};
    \node at (3.3,3) {2};
    \draw[thick] (0,0) -- node[below] {$G_3$}  (3,0) ; 
    \draw[thick] (0,3) -- node[left] {$G_4$} (0,0); 
    \draw[thick] (3,3) -- node[right] {$G_2$} (3,0); 
    \draw[thick] (0,3) -- node[above] {$G_1$} (3,3);
    \draw[thick] (0,0) -- node[above left=0.25cm and 29pt,anchor=south] {$G_5$} (3,3); 
    \draw[thick] (0,3) -- node[above right=0.25cm and 27pt,anchor=south] {$G_6$} (3,0);
    \end{tikzpicture}
    \caption{4-node graph for randomized EFGs. Each node represents a player and each edge represents a game $G_i$ between the corresponding players.}
    \label{fig:graph1}
\end{figure}
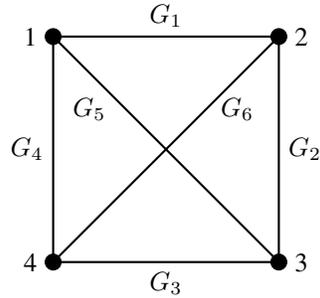

\textbf{Kuhn Poker.} Kuhn poker is a simplified version of poker proposed by \cite{kuhn1950poker}. The deck contains only three cards, namely Jack, Queen and King. Each player is dealt one card, and the third is left unseen. Player 1 can either check or bet, and subsequently Player 2 can also either check or bet. Finally, if Player 1 checks in round 1 and Player 2 bets in round 2, Player 1 gets another round to fold or call. Eventually, the player with the highest card wins the pot. In the sequence form representation of the game, Kuhn poker has dimension $\vert\X\vert \times \vert\X\vert = 13 \times 13$ and the corresponding payoff matrix can be easily computed by hand. For the simulation we show in Figure \ref{fig:experiments}, we run an experiment with 5 players on a graph where each player plays in exactly two Kuhn poker games with randomized initial conditions. This limitation was set in order to reduce the convergence time, since empirically we observe that increasing the number of players greatly increases the convergence times. 
\begin{figure}
    \centering
    \includegraphics[width=0.8\textwidth]{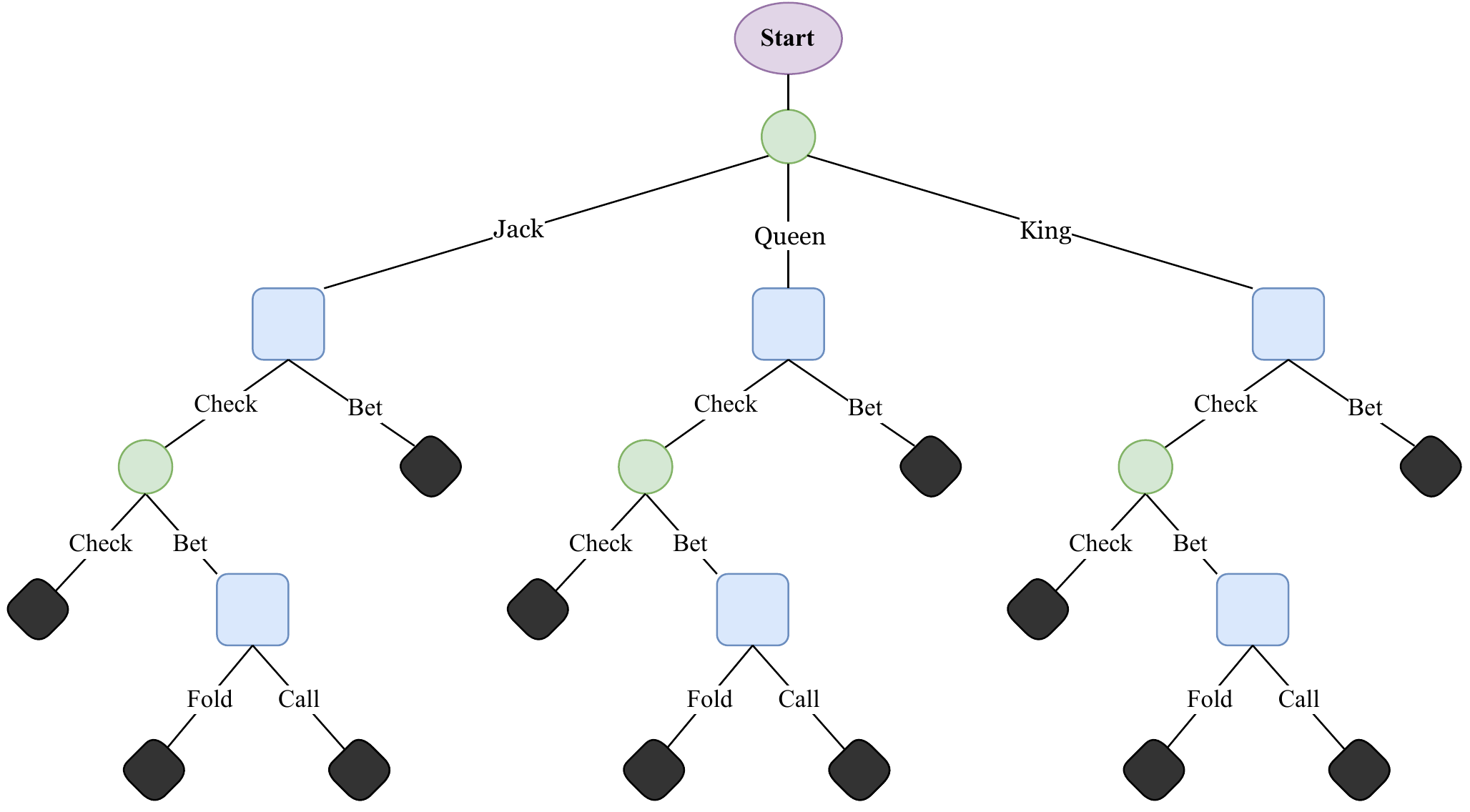}
    \caption{Extensive form representation of Kuhn poker from the perspective of one player. The blue nodes represent decision points for the player, green nodes represent observation points (either the player observes their card or the other player takes an action) and finally the black nodes denote the terminal states of the game.}
    \label{fig:my_label}
\end{figure}



\textbf{A note on scaling.} An empirical observation from our simulations is that the number of nodes in the network as well as the sparsity of the graph plays a major role in convergence times, particularly the \emph{last-iterate} convergence times. This intuitive observation presents an interesting challenge when modeling truly large-scale problems. For instance, a setting such as Texas Hold'em poker admits a huge number of parameters (of order $10^{18}$). Even in the two-player case this is prohibitively large, and this issue is compounded if we are in the multiplayer setting. 
As an illustrative example, consider a network game where every agent plays the ubiquitous zero-sum game, Matching Pennies, against two other players. Figure \ref{fig:mp} shows that the convergence times drastically increase when we go from a 4-node graph to a 20-node graph. Similarly, in our experiments with extensive form games in sequence form, it becomes difficult to simulate larger games (such as Leduc poker, which has dimension $\vert\X\vert \times \vert\X\vert = 337$) once there are multiple players playing in several games. This is a practical limitation which represents an interesting divide between our theoretical results and the reality of many large-scale, real world games. It is certainly a fascinating research direction to find ways to bridge this gap in future research.

\begin{figure}[ht]
\centering
\begin{minipage}{.48\textwidth}
  \centering
  \includegraphics[width=.95\linewidth]{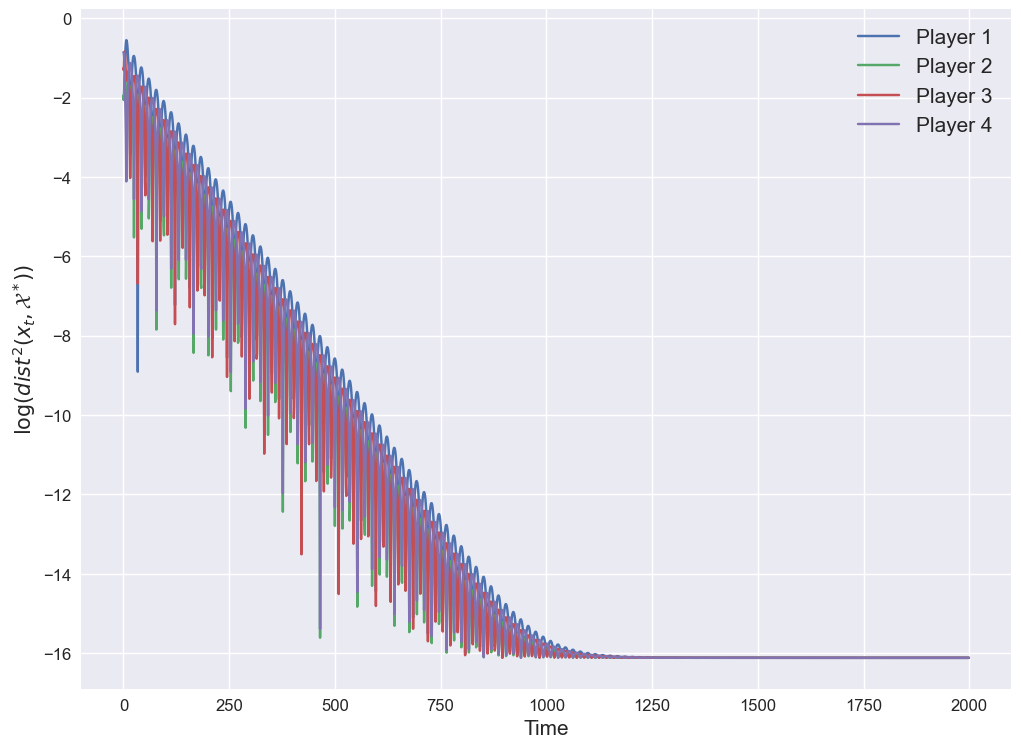}
\end{minipage}%
\hfill
\begin{minipage}{.48\textwidth}
  \centering
  \includegraphics[width=.95\linewidth]{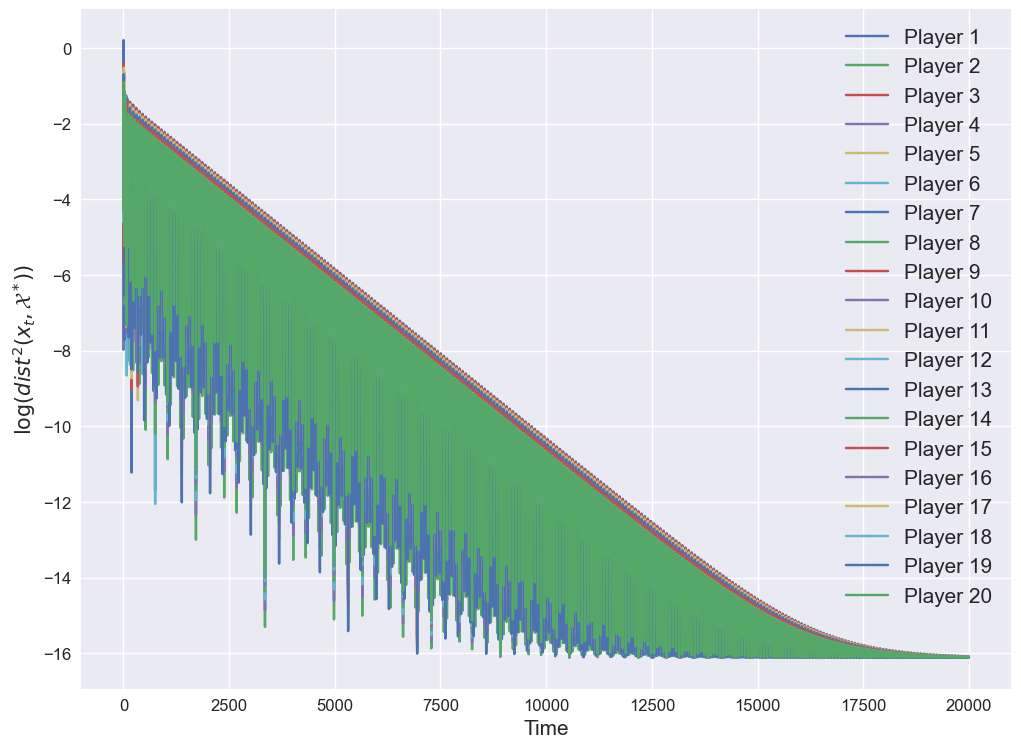}
\end{minipage}
\caption{Simulations using OGA in network Matching Pennies games. (Left) Convergence times for 4-player game; (Right) Convergence times for 20-player game.}
\label{fig:mp}
\end{figure}

\end{document}